%% file: GSED_Complexity.tex
\documentclass[a4paper,12pt]{article}
\usepackage{prelude}
\usepackage{quantum}
\usepackage{amsmath}
\usepackage{amssymb}
\usepackage{epsfig}
\usepackage{graphicx}
\usepackage[dvipsnames]{xcolor}
\usepackage{tikz}
\graphicspath{ {Images/} }
\numberwithin{equation}{section} 
\usepackage{authblk}
\usepackage{xparse}
\usepackage{cleveref}
\usepackage[backend=bibtex8,style=alphabetic,doi=false,isbn=false,url=false,maxbibnames=5,sorting=nty,sortcites=false]{biblatex}

\newcommand\nbd\nobreakdash

\newcommand{\cD}{\mathcal{D}}
\newcommand{\cU}{\mathcal{U}}

\DeclareMathOperator{\deficit}{deficit}
\DeclareMathOperator{\sdeficit}{square\_deficit}
\DeclareMathOperator{\odeficit}{obstruction\_deficit}
\DeclareMathOperator{\tdeficit}{total\_deficit}

\def\R{\cal{R}}

\DeclareMathOperator{\Er}{\mathcal{E}_{\rho}}

\newcommand{\YES}{{{\normalfont\textsc{Yes}}}}
\newcommand{\NO}{{{\normalfont\textsc{No}}}}

\newcommand{\GSED}{{\normalfont\textsf{GSED}}}
\newcommand{\FGSED}{{\normalfont\textsf{FGSED}}}
\newcommand{\NEEXP}{{\normalfont\textsf{NEEXP}}}
\newcommand{\FP}{{\normalfont\textsf{FP}}}
\newcommand{\FEXP}{{\normalfont\textsf{FEXP}}}
\newcommand{\EXPNEXP}{{\normalfont\textsf{EXP\textsuperscript{\textsf{NEXP}}}}}

\newcommand{\PNEEXP}{{\normalfont\textsf{P\textsuperscript{\textsf{NEEXP}}}}}
\newcommand{\PSPACE}{{\normalfont\textsf{PSPACE}}}

\newcommand{\NEXP}{{\normalfont\textsf{NEXP}}}
\newcommand{\NP}{{\normalfont\textsf{NP}}}
\newcommand{\EXP}{{\normalfont\textsf{EXP}}}

\newcommand{\EXPGSED}{{\normalfont\textsf{EXP\textsuperscript{\textsf{GSED}}}}}
\newcommand{\QMA}{{\normalfont\textsf{QMA}}}
\newcommand{\QMAEXP}{{\normalfont\textsf{QMA}\textsubscript{\textsf{EXP}}}}
\newcommand{\QMAEEXP}{{\normalfont\textsf{QMA}\textsubscript{\textsf{EEXP}}}}
\newcommand{\EXPQMAEXP}{{\normalfont\textsf{EXP\textsuperscript{\textsf{QMA}\textsubscript{\textsf{EXP}}}}}}
\newcommand{\PQMAEXP}{{\normalfont\textsf{P\textsuperscript{\textsf{QMA}\textsubscript{\textsf{EXP}}}}}}

\bibliography{bibliography}

\begin{document}
	\title{Computational Complexity of the Ground State Energy Density Problem }
	\author[1]{James D. Watson}
	\author[1]{Toby S. Cubitt}
	\affil[1]{Department of Computer Science, University College London, UK}

	\date{}

	\maketitle

	\begin{abstract}
          We study the complexity of finding the ground state energy density of a local Hamiltonian on a lattice in the thermodynamic limit of infinite lattice size.
          We formulate this rigorously as a function problem, in which we request an estimate of the ground state energy density to some specified precision; and as an equivalent promise problem, \GSED, in which we ask whether the ground state energy density is above or below specified thresholds.

          The ground state energy density problem is unusual, in that it concerns a single, fixed Hamiltonian in the thermodynamic limit, whose ground state energy density is just some fixed, real number.
          The only input to the computational problem is the precision to which to estimate this fixed real number, corresponding to the ground state energy density.
          Hardness of this problem for a complexity class therefore implies that the solutions to all problems in the class are encoded in this single number (analogous to Chaitin's constant in computability theory).

          This captures computationally the type of question most commonly encountered in condensed matter physics, which is typically concerned with the physical properties of a single Hamiltonian in the thermodynamic limit.
          We show that for classical, translationally invariant, nearest neighbour Hamiltonians on a 2D square lattice, $\PNEEXP\subseteq\EXPGSED\subseteq \EXPNEXP$, and for quantum Hamiltonians $\PNEEXP\subseteq\EXPGSED\subseteq \EXPQMAEXP$.
          With some technical caveats on the oracle definitions, the \EXP{} in some of these results can be strengthened to \PSPACE.
          We also give analogous complexity bounds for the function version of \GSED.

	\end{abstract}

	\pagebreak

	\tableofcontents

	\pagebreak

	\section{Introduction}
        \input{intro}

	\section{Main Results}\label{sec:main_results}
	Define the energy density of the finite lattice as
\begin{definition}[Ground State Energy Density]
	Consider a translationally invariant Hamiltonian defined on an $L\times H$ lattice, $H^{\Lambda(L \times H)})$.
	The \emph{ground state energy density} is defined as
	\begin{equation}
	\Er(L,H) := \frac{\lambda_0(H^{\Lambda(L \times H)})}{LH}.
	\end{equation}
	The \emph{thermodynamic limit} of the ground state energy density is defined as the limiting value as the lattice width and height bare taken to infinity:
	\begin{equation}
	\Er:= \lim\limits_{L,H\rightarrow\infty} \Er(L,H).
	\end{equation}
	If the ground state energy density is referred to without qualification, then it is referring to the thermodynamic limit case.
\end{definition}
This limit is well defined \cite{Cubitt_PG_Wolf_Undecidability}.
We now consider some useful definitions for the computational problems.
For all these definitions we will be referring to the infinite lattice case.

We can cast the problem of finding $\Er$ as a computational promise problem similar in spirit to the local Hamiltonian problem:
\begin{samepage}
  \begin{definition}[Ground State Energy Density (\GSED) promise problem]
    \label{Def:Decision_Problem_GSED_2}
    \textbf{Problem Parameters:} A fixed, translationally invariant, nearest-neighbour Hamiltonian on a $2D$ infinite square lattice of $d$-dimensional spins. \\
    \textbf{Input:} Two real numbers $\beta$ and $\alpha$, such that $\beta - \alpha = \Omega(2^{-p(n)})$, for some integer $n$ and polynomial $p(n)$.\\
    \textbf{Output:} Determine whether $\Er>\beta$ (\NO{} instance) or $\Er<\alpha$ (YES instance). \\
    \textbf{Promise:} The ground state energy density does not lie between in the interval $[\alpha,\beta]$.
  \end{definition}
\end{samepage}

This is perhaps more naturally thought of in terms of the corresponding function problem:
\begin{definition}[Ground State Energy Density (\FGSED) function problem]
  \label{Def:GSED_Error_Function_Definition}
  \textbf{Problem Parameters:} A fixed, translationally invariant, nearest-neighbour Hamiltonian acting on an $2D$ infinite lattice of $d$-level spins. \\
  \textbf{Input:} An error bound $\epsilon$, specified in binary.  \\
  \textbf{Output:} An approximation to the ground state energy density, $\tilde{\Er}$ such that  $|\Er - \tilde{\Er}| \leq \epsilon$.
\end{definition}

The promise and function problems are equivalent up to $\log$-space computation, by standard binary search arguments.

We will often restrict \GSED{} in \cref{Def:Decision_Problem_GSED_2} to classical Hamiltonians, rather than general (quantum) Hamiltonians.
When we wish to highlight this distinction, we refer to these as \emph{classical \GSED} and \emph{quantum \GSED}, respectively.


	The main results of this work are as follows:
	\begin{theorem} \label{Theorem:NEXP_Hardness}
		$\PNEEXP\subseteq\EXPGSED\subseteq\EXPNEXP${} for classical \GSED.
	\end{theorem}
	Here \NEEXP{} is defined analogously with \NP, but the verifying TM is allowed doubly exponential time to run and the witness can be doubly exponentially long.
	We expect that the \EXPNEXP{}  upper bound presented here is tight and there is potentially room to improve the lower bound.
	The above theorem implies:
	\begin{corollary}\label{Corollary:NEEXP-Hardness}
		\GSED{} is \NEEXP-hard under exponential time Turing reductions, for a classical, translationally invariant, nearest-neighbour Hamiltonian.
	\end{corollary}
	We also prove:
	\begin{theorem}\label{Theorem:Results:GSED_NEXP}
		Classical $\GSED\in\NEXP$.
	\end{theorem}
	\Cref{Corollary:NEEXP-Hardness} and \cref{Theorem:Results:GSED_NEXP} are not in conflict with each other.
	Allowing exponential-time Turing reductions (as opposed to the polytime Turing reductions usually considered) allows exponentially harder problems to be solved.

	The fact we are considering \EXPGSED{} rather than \GSED{} with polytime reductions is fundamental to the problem being about estimating the the ground state energy density for \emph{a particular Hamiltonian}, where the problem instances differ only in the precision to which that same ground state energy density should be computed
    (rather than each problem instance corresponding to a different Hamiltonian).
    We show that, using our hardness construction, one should not expect $\NP\subseteq \mathsf{P}^{\GSED}$ unless the polynomial hierarchy collapses to $\Sigma_2^p$.

	We can also consider the case of quantum Hamiltonians:
	\begin{theorem}
          \PNEEXP$\subseteq$\EXPGSED$\subseteq$\EXPQMAEXP for quantum \GSED{}.
	\end{theorem}

	For the function problem, one readily obtains the corresponding complexity bounds:
	\begin{theorem}
		$\FGSED \in \FP^{\NEXP}$ for classical $\FGSED$.
	\end{theorem}
	We also get the bound
	\begin{lemma}
		$\FP^{\NEEXP} \subseteq \FEXP^{\FGSED}\subseteq \FEXP^{\NEXP}$, for $\FGSED$ for a fixed classical, translationally invariant, nearest neighbour Hamiltonian.
	\end{lemma}

	\section{Preliminaries}\label{sec:preliminaries}
	Let $\mathcal{B}(\mathcal{H})$ be the space of bounded linear operators on a complex Hilbert space $\mathcal{H}$.
	Define $\Lambda(L \times H) := \{1,\dots, L\} \times \{1,\dots,H\}$ to be the square lattice of width $L$, height $H$, with $L,H\in \mathbb{N}$.
	We attach to each site $i\in \Lambda(L \times H)$ in the lattice a Hilbert space $\mathcal{H}_i \cong \mathbb{C}^d$.
	Given a string $x\in \{0,1\}^n$, then $|x|=n$ will denote the binary length of the string.
	For a given Hamiltonian $H$, we will denote its eigenvalues as $\lambda_i(H)$, such that $\lambda_0(H)\leq\lambda_1(H)\leq\lambda_2(H) \leq \dots$.

	Given a lattice $\Lambda(L\times H)$, a Hamiltonian $H=\sum_{i} h_i$ is \emph{nearest-neighbour} if $h_i\in \mathcal{B}(\mathbb{C}^d\otimes \mathbb{C}^d)$ such that each $h_i$ acts non-trivially only on neighbouring pairs of lattice sites.
	We write the interaction between neighbouring sites as $h_{\langle i,j\rangle}$.
	Furthermore, \emph{translational invariance} implies $h_{\langle i,j\rangle}=h\in\mathcal{B}(\mathcal{H})$ for any $i,j$.
	By a \emph{classical Hamiltonian}, we mean a Hamiltonian which is diagonal in the standard basis.
	To distinguish general Hamiltonians from classical Hamiltonians we will often call them \emph{quantum Hamiltonians}.

	\subsection{Complexity Classes}

	\begin{definition} \emph{\NEXP{} or \textsf{NEXPTIME}}

		A language $L$ is in \NEXP{} if there exists a positive constant $k$ and a deterministic Turing Machine $M$ such
		that for each instance $x$ and a classical witness $w$ such that $|w| = O(2^{|x|^k})$, on input $(x, w)$, $M$ halts in $O(2^{|x|^k})$ steps and
		\begin{itemize}
			\item if $x \in L,$  $\exists w$ such that $M$ accepts $(x, w)$ with probability 1.
			\item if $x \not\in L$ then $\forall w$, $M$ accepts $(x, w)$ with probability 0.
		\end{itemize}
	\end{definition}
	\noindent
	We note that \NEXP = \textsf{NTIME}$(2^{cn})$ \cite{Papadimitriou_1994}.
	\begin{definition} \emph{\NEEXP{} or \textsf{N2EXP}}

		A language $L$ is in \NEEXP \space if there exists a positive constant $k$ and a deterministic Turing Machine $M$ such
		that for each instance $x$ and a classical witness $w$ such that $|w| = O(2^{2^{|x|^k}})$, on input $(x, w)$, $M$ halts in $O(2^{2^{|x|^k}})$ steps and
		\begin{itemize}
			\item if $x \in L,$  $\exists w$ such that $M$ accepts $(x, w)$ with probability 1.
			\item if $x \not\in L$ then $\forall w$, $M$ accepts $(x, w)$ with probability 0.
		\end{itemize}
	\end{definition}

	\begin{definition} \emph{\NEEXP{} or \textsf{N2EXP}}
	
	A language $L$ is in \NEEXP \space if there exists a positive constant $k$ and a deterministic Turing Machine $M$ such
	that for each instance $x$ and a classical witness $w$ such that $|w| = O(2^{2^{|x|^k}})$, on input $(x, w)$, $M$ halts in $O(2^{2^{|x|^k}})$ steps and
	\begin{itemize}
		\item if $x \in L,$  $\exists w$ such that $M$ accepts $(x, w)$ with probability 1.
		\item if $x \not\in L$ then $\forall w$, $M$ accepts $(x, w)$ with probability 0.
	\end{itemize}
\end{definition}
	We also define $\mathsf{QMAEEXP}$ the same way as $\mathsf{QMA}$, but allowing for a doubly-exponentially long witness and circuit runtime.

	Throughout, we will make use of oracle classes: these are the set of problems solvable by a Turing Machine with access to an oracle solving some problem (or class of problems).
	\begin{definition}[Oracle Turing Machines \cite{Arora_Barak_2010}]
		An oracle Turing machine is a TM, $M$, that has a special read/write tape we call $M$'s oracle tape and three special states $q_{query}$, $q_{yes}$, $q_{no}$.
		To execute $M$, we specify in addition to the input a language
		$O \subset \{0, 1\}^*$ that is used as the oracle for $M$.
		Whenever during the execution $M$ enters the state
		$q_{query}$, the machine moves into the state $q_{yes}$ if $q \in O$ and $q_{no}$ if $q \not\in O$, where $q$ denotes the contents of the special oracle tape.
		 Note that, regardless of the choice of $O$, a membership query to $O$
		counts only as a single computational step.
		If $M$ is an oracle machine, $O \subset \{0, 1\}^{*}$ a language, and
		$x \in \{0, 1\}^*$, then we denote the output of M on input x and with oracle $O$ by $M^O(x)$.
	\end{definition}

	\begin{definition}[Oracle Classes \cite{Arora_Barak_2010}]
		For every $O \subset \{0, 1\}^{*}$, $\mathsf{P}^O$ is the set of languages decided by a polytime deterministic
		TM with oracle access to $O$ and \NP$^O$ is the set of languages decided by a polytime nondeterministic TM with oracle access to $O$.
		Similarly for $\PSPACE^O$ and
		$\EXP^O$.
	\end{definition}
	For the particular case of $\PSPACE^O$ machines, the \PSPACE{} machine can execute exponentially many computational steps.
        So there is a subtlety as to whether the space bound also applies to the oracle tape or not.
        Multiple possible definitions for what the $\PSPACE$ machine has access to with regards to the oracle tape have been considered in the literature~\cite{Fortnow_94}.
	We discuss the different results we get depending on the choice of definition in \cref{Sec:PSPACE_Results}.

	\begin{definition}[Oracle Function Classes \cite{Papadimitriou}]
		For every $O \subset \{0, 1\}^{*}$, $\FP^O$ is the set of functions $f:\{0,1\}^*\rightarrow \{0,1\}^*$ that can be computed by a polytime deterministic TM with oracle access to $O$.
		$\FEXP^O$ is similarly the set of functions computed by an exponential time deterministic TM with oracle access to $O$.
	\end{definition}

	\section{Tiling Preliminaries}\label{sec:tiling_preliminaries}
	Wang tilings will play a central role in this work.
	\begin{definition}[Wang Tiles]
		\emph{Wang tiles} are unit length square tiles with markings on each of the four edges.
		For a given  set of Wang tiles $\{t_i\}_{i=1}^n$, the markings define horizontal matching rules $\mathcal{R}_{Horz}$ (respectively, vertical matching rules $\mathcal{R}_{Vert}$) such that two tiles $t_i,t_j$ can only be placed next to each other horizontally (vertically) if $(t_i,t_j)\in \mathcal{R}_{Horz}$ $((t_i,t_j)\in \mathcal{R}_{Vert})$.
	\end{definition}
	\noindent We now consider specific sets of Wang tiles that we will employ throughout this work.

	\subsection{Robinson Tiles}

	Robinson's tiling~\cite{Robinson_1971} is based on a set of 5 basic tiles, shown in figure \cref{Fig:Robinson_Tiles}, with the rule that one tile can be placed next to another only if the arrow heads on the first tile correctly join with the arrow tails on the adjacent tile.
	I.e.\ the tiling rules enforce the condition that \emph{arrow heads on one tile must meet arrow tails of the same type on its neighbour in the appropriate direction}.

	\begin{figure}[h!]
		\centering
		\includegraphics[width=0.9\textwidth]{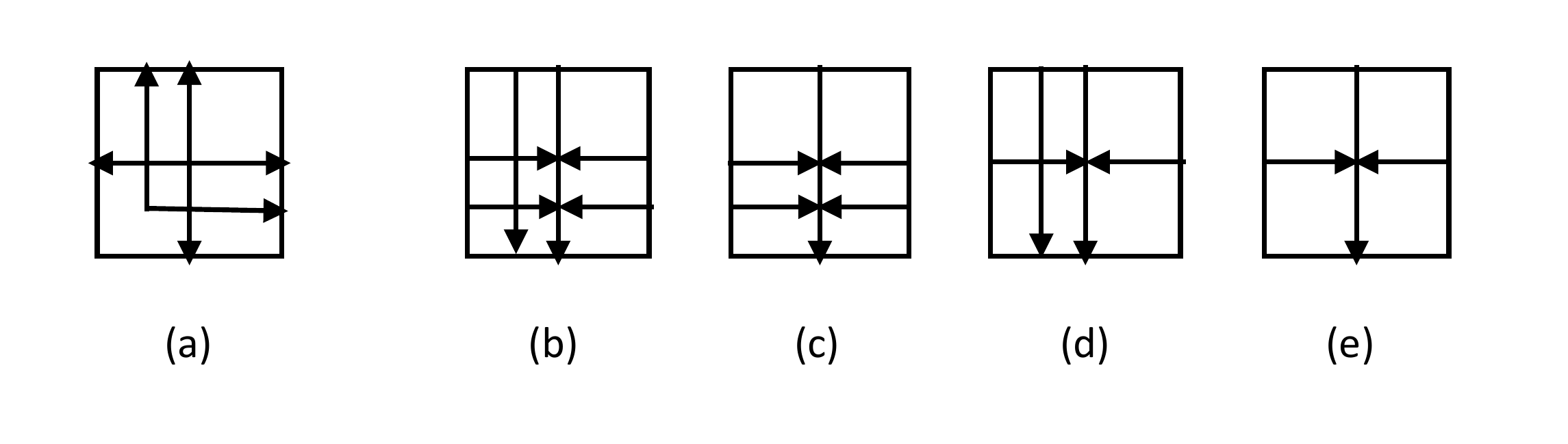}
		\caption{The five Robinson tiles we will use.
		Image taken from \cite{Cubitt_PG_Wolf_Undecidability}. }
		\label{Fig:Robinson_Tiles}
	\end{figure}

	Tile (a) in \cref{Fig:Robinson_Tiles} has arrows on all sides of the tile and is known as a \emph{cross} and in this depiction is said to face up and to the right.
	The other 4 tiles are known as \emph{arms}.
	Each of the arms has a principle arrow across the centre of the tile and which indicates its direction (all the tiles depicted in \cref{Fig:Robinson_Tiles} are facing downwards).
	Arrow markings can be either red or green.
	On a given arm the horizontal and vertical arrows must have different colours and on cross tiles we force all arrow markings to have the same colour.
        The Robinson tile set includes all rotations and reflections of these basic tiles.

	When these tiles are augmented with certain additional markings, described in~\cite{Robinson_1971,Cubitt_PG_Wolf_Undecidability}, the tiling rules force a pattern of interlocking, nested squares to form in any valid tiling of the plane (see \cref{Fig:3-7-Square}(c)).
	The series of squares have side lengths $3,7,15,31, \dots, 2^n-1$, for $n\in \mathbb{N}$ (see \cref{Fig:Big-Square}).
	Robinson adds additional coloured markings to the tiles, such that for odd $n$ the borders formed by the double-arrow tile markings running along the edges of the squares are green, and for even $n$ they are red.
	We direct the reader to \cite{Robinson_1971} and \cite{Cubitt_PG_Wolf_Undecidability} for more detailed discussions of tiling pattern and how it is formed.
		\begin{figure}[h!]
		\centering
		\includegraphics[width=1\textwidth]{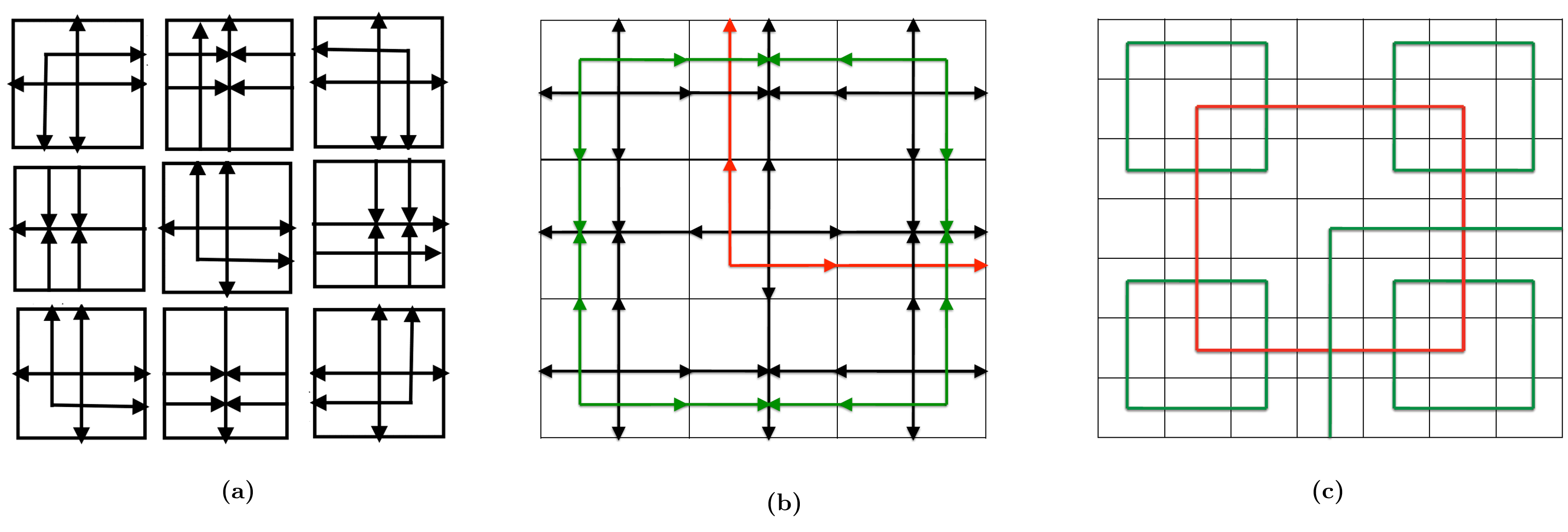}
		\caption{ (a) A possible tiling arrangement to create a 3-square.
			(b) shows the same square once the coloured arrows have been introduced.
			(c) shows a 7-square having combined several 3-squares.
			Images (b) and (c) taken from \cite{Cubitt_PG_Wolf_Undecidability}. }
		\label{Fig:3-7-Square}
	\end{figure}

	\begin{figure}[h!]
		\centering
		\includegraphics[width=0.7\textwidth]{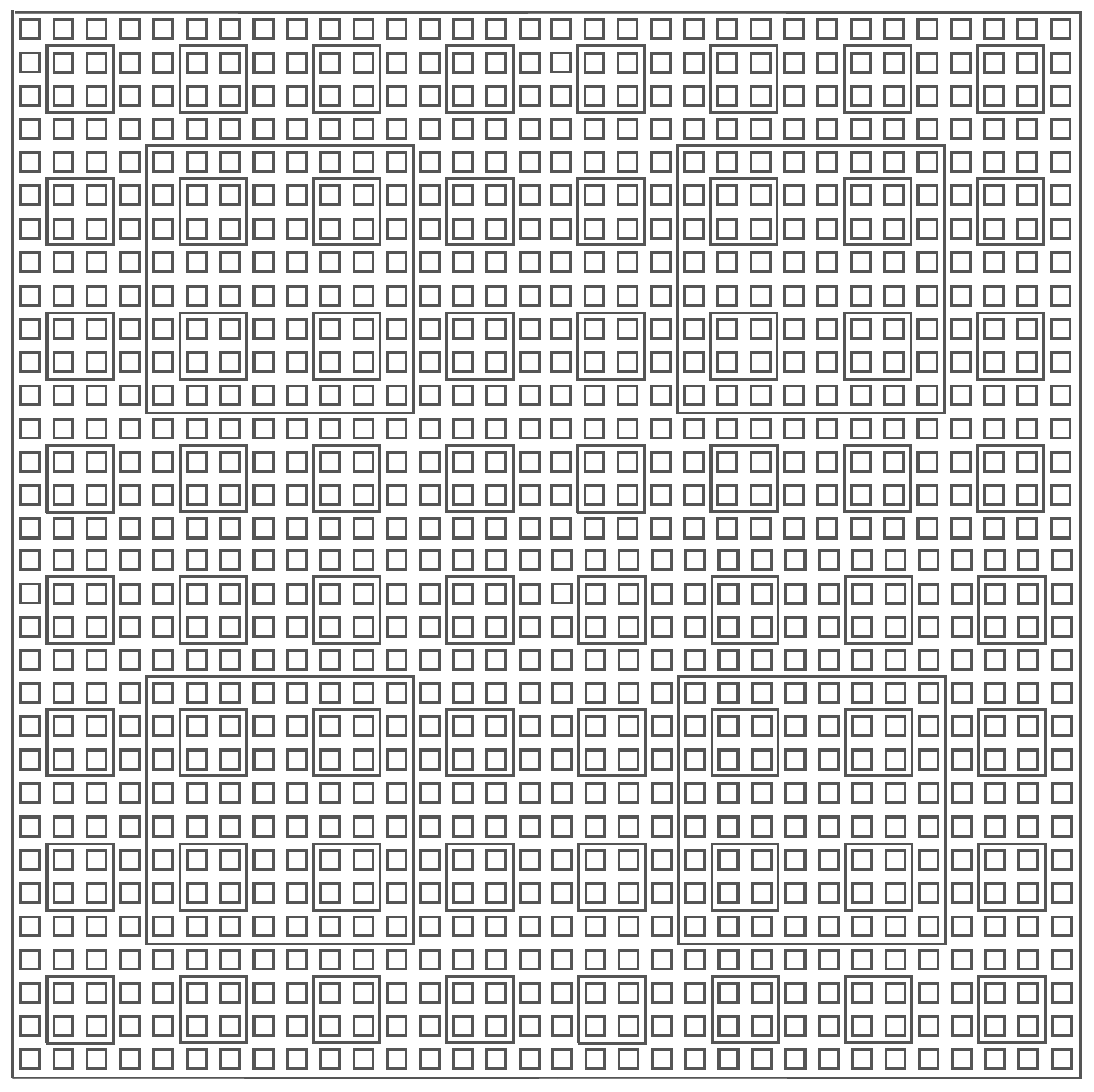}
		\caption{ A Robinson tiling pattern showing only red borders.
			Image modified from \cite{Cubitt_PG_Wolf_Undecidability}.}
		\label{Fig:Big-Square}
	\end{figure}

	For our purposes we will mostly focus on red borders, and refer to these as just \emph{borders}.
        The interior of the border is referred to as a \emph{square}.
        We refer to a red border of side $4^n-1$ as an \emph{$n$-border}.
        When green borders are referenced, this will always be made explicit.

	Consider \cref{Fig:Robinson_Tiles+Dashes}.
	Let $R_v^{r}, R_v^{l}, R_h^{u}, R_h^{d}$ be the sets of Robinson tiles which contain tiles of type $(b)$, $(c)$, and $(d)$ markings, where the double-arrow markings going across the entire tile are red, and where the arrow markings going across the entire tile are respectively facing right, left, up or down.
	Let $R_X$ be the set of red crosses, and let $R_X^{UR},R_X^{UL}, R_X^{DL}, R_X^{DR}$ be the cross tiles that have double arrow markings facing up-right, up-left, down-left and down-right respectively.

	\begin{definition}[$n$-borders]\label{def:n-border}
	Consider a $(4^n-1) \times (4^n-1)$ subset of a tiling grid, not including its interior.
	Then the region forms $2$-border if for every point along the left vertical edge, right vertical edge, bottom horizontal edge, and along the top horizontal edge satisfies $\Lambda(p)\in R_h^{l},R_h^{r},R_h^{d},R_h^{u}$, respectively.
	Furthermore, the tile in the top-right corner $R_X^{DL}$, top-left corner is $R_X^{DR}$, bottom-left corner is $R_X^{UR}$, and bottom right corner is $R_X^{UL}$.
	\end{definition}

	Finally note that Robinson tiles allows for two half-planes to be translated relative to each other without violating any of the tiling rules.
	We wish to avoid this and hence use the modified set of Robinson tilings introduced in \cite{Cubitt_PG_Wolf_Undecidability}, such that the final set of tiles is all rotations and reflections of those shown in \cref{Fig:Robinson_Tiles+Dashes}.
	It is shown in \cite{Cubitt_PG_Wolf_Undecidability} that these tiles produce the same pattern of nested squares, but prevent any two half-planes from be translated relative to each other.

	\begin{figure}[h!]
		\centering
		\includegraphics[width=0.8\textwidth]{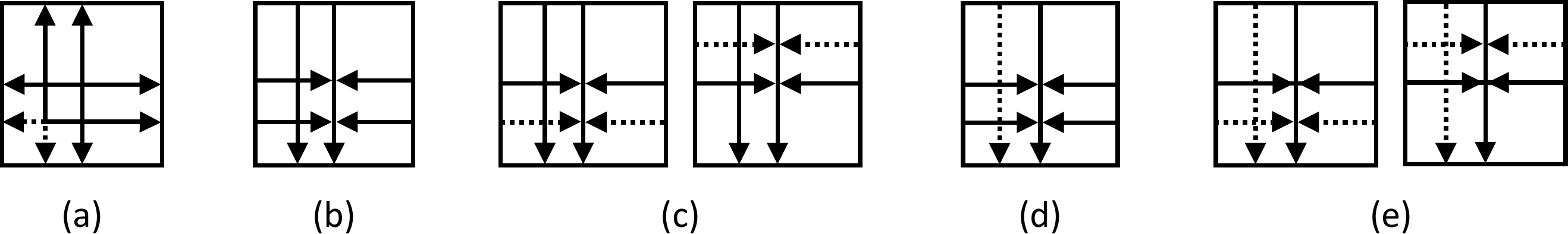}
		\caption{The standard Robinson tiles with additional dashed markings added in to prevent slippage between planes.
		Image modified from \cite{Cubitt_PG_Wolf_Undecidability}.}
		\label{Fig:Robinson_Tiles+Dashes}
	\end{figure}

	\subsection{Encoding Turing Machines with Tiles}
        It is well known that the evolution of a classical Turing Machine can be encoded as a set of Wang tiles \cite{Berger_1966,Robinson_1971}.
        To see this, consider a particular TM.
        The TM tape at a particular time step is a set of tape cells with symbols written in them, where one particular cell has the TM head over it.
        The TM will then evolve deterministically according to its transition rules.

        Now consider an $L\times L$ tiling grid.
        It is possible to construct a set of Wang tiles such that the tiling pattern simulates the TM's evolution for $L$ steps.
        The tile set is chosen to be tiles with all possible combinations of Turing Machine tape cell markings, plus TM head and state markings.
        The evolution of the TM can then be encoded as a tiling of a square lattice, where rows of tiles represent the configuration of the TM tape, together with the head location and current internal state, at a particular time step.
        Adjacent rows encode the TM configuration at successive time steps.
        The correct TM evolution is then enforced by tiling rules.
        (See \cref{Fig:TM_Encoded_As_Tiles} for an example of such an encoding, and see \cite{Berger_1966, Robinson_1971, Gottesman-Irani, Cubitt_PG_Wolf_Undecidability, Bausch2015} for some further detailed discussions on this topic.)

		\begin{figure}[t]
			\centering
			\includegraphics[width=0.6\textwidth]{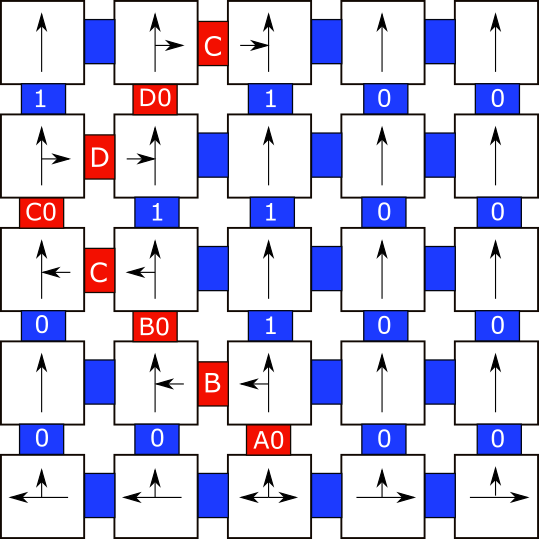}
			\caption{The evolution of a classical TM can be represented by Wang tiles, where colours of adjacent tiles have to match, and arrow heads have to meet arrow tails.
				Here the evolution runs from the bottom of the square to the top.
				The red labels between adjacent rows represent the position and state of the TM head, and the red labels between adjacent columns represent movement of the TM head after it has acted on the cell.}
			\label{Fig:TM_Encoded_As_Tiles}
		\end{figure}

		\subsection{Encoding Turing Machines in the Robinson Tiling} \label{Sec:Encode_TM}

		In this section we review how the tiling-encoding of TMs can be combined with the Robinson tiling to create a new set of tiles which, when the plane is tiled according to the tiling rules, encodes the evolution of a separate TM within each $n$-square in the Robinson tiling pattern.
		This construction was introduced in \cite{Robinson_1971} to prove undecidability of the tiling of a 2D plane.

		Encoding the evolution of a TM directly within the interior of a $n$-border is not possible as the Robinson tiling pattern is composed of $m$-squares nested within other $n$-squares, $m<n$.
		Thus TMs would overlaps with each other.
		\cite{Robinson_1971} circumvents this problem by identifying a sub-grid within each Robinson $n$-border which allows a TM to be encoded without overlapping with the smaller $m$-squares, $m<n$, nested within.

		\begin{definition}[Free Rows/Columns and Free Squares, \cite{Robinson_1971}] \label{Def:Free_Rows/Columns}

			A \emph{free row/column} of square  is a row/column in a Robinson $n$-border that stretches across the border's interior uninterrupted by any of the $m$-borders with $m<n$.

			A \emph{free square or tile} is a square in the grid that is both in a free row and a free column.
			Within an $n$-square there are exactly $2^n +1$ free rows/columns.
		\end{definition}

		\begin{lemma}[Encoding TM in Robinson Tiling, \cite{Robinson_1971}] \label{Lemma:UTM_on_Robinson_Square}
		Consider any classical Turing Machine which can have its evolution be encoded in a $(2^n+1)\times (2^n+1)$ grid of Wang tiles.
		Then the evolution of this TM
		can be encoded in the free rows and columns of an $n$-square in a Robinson Tiling.

	\end{lemma}

	\noindent We will use the details of Robinson's construction of \cref{Lemma:UTM_on_Robinson_Square} later, hence we provide some exposition here.

	Consider a Robinson $n$-border.
	Following \cite{Robinson_1971}, to demarcate where the free tiles are so that we can encode a Turing Machine in them, introduce a new kind of marking on the tiles called an `obstruction signal'.
	These signals are designed so they are emitted and absorbed from the outside of a red border and while also being absorbed by the inside of a border, as seen in \cref{Fig:Obstruction_Signals}.
	In terms of tiles, these markings are formed by adding an additional set of markings such that tiles of type (b) in \cref{Fig:Robinson_Tiles+Dashes} with a red double-arrow ``emit'' the obstruction signals from one side and ``absorb'' them on both sides.
	Tiles that do not emit or absorb obstruction signals force them to propagate in the same direction.
	The obstruction signals are only emitted from the outer edges of a red Robinson border.
	A \emph{free tile} is one which does not have an obstruction signal going across it in either direction.
	In our new tile set, we only encode the Turing Machine tape, head and state symbols in the free tiles.

	\begin{figure}[h!]
		\centering
		\includegraphics[width=0.5\textwidth]{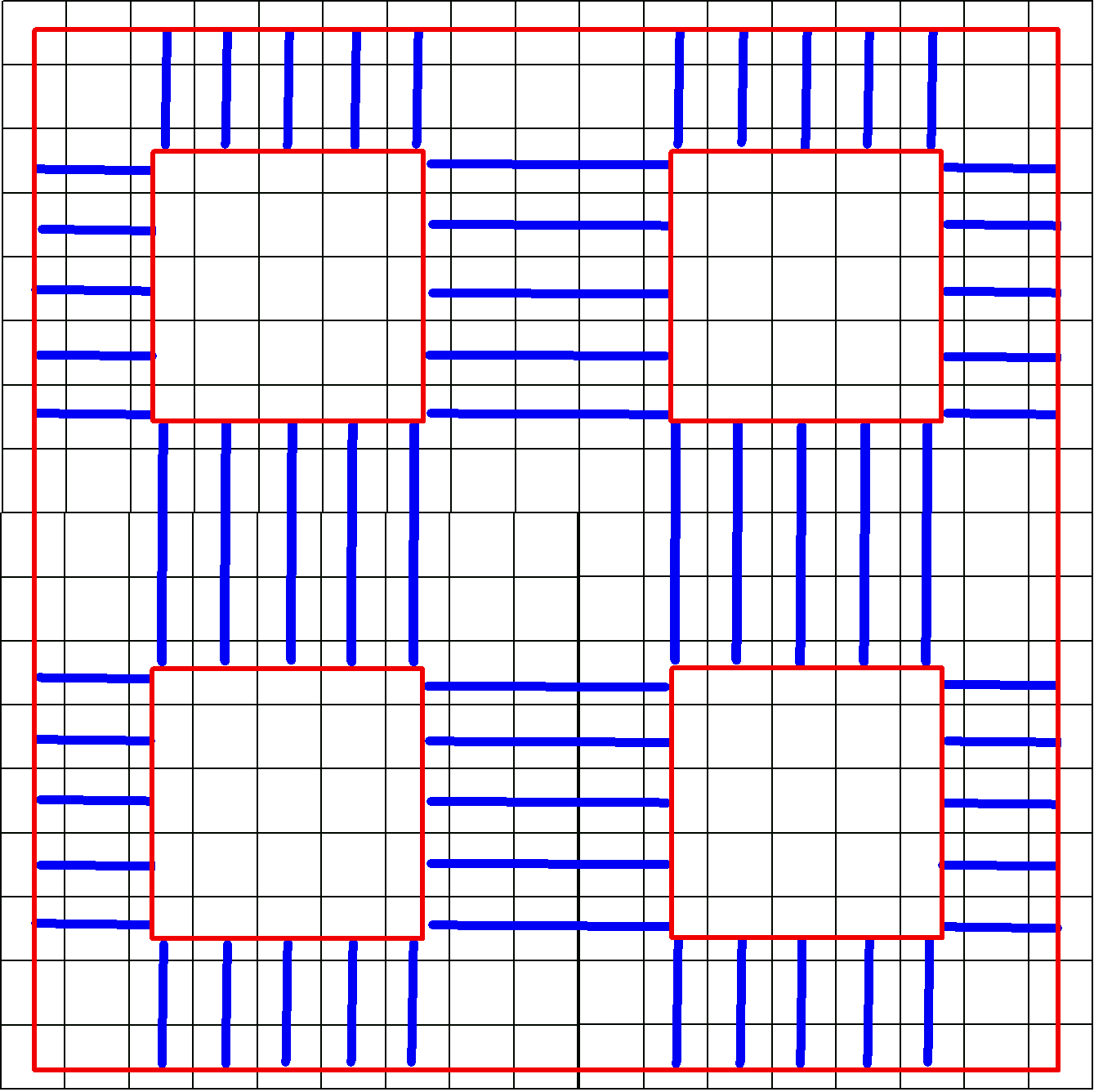}
		\caption{ The obstruction signals for a red $2^4$-square are shown in blue. Each of the tiles within the $2^2$-squares emits a signal outwards. The free rows are the rows in which there are no obstruction signals running horizontally (for example the central row). The free columns are the columns in which there are no obstruction signals running vertically (for example the central column).  }
		\label{Fig:Obstruction_Signals}
	\end{figure}

	\paragraph{Transmitting Signals between Free Tiles}

	Thus we are able to encode the evolution of a Turing Machine in these free tiles, effectively creating a $(2^n + 1)\times (2^n +1)$ square for it to run in.
	There is a problem in that the free tiles are not spatially close to each other.
	To solve this, \cite{Robinson_1971} implicitly introduces a new set of tile markings: Turing Machine signals.
	These signals can be emitted and absorbed by free tiles and run along free rows and columns.
	Otherwise they are absorbed by tiles with double arrowed red markings: tiles of types (a), (b), or (c), shown in \cref{Fig:Robinson_Tiles+Dashes}, on the sides of tiles parallel to the red double arrow lines.
	Tiles which are not free tiles, and do not absorb the TM signals, force the TM signals to propagate across them.
	These signal markings allow the tiling to transmit the necessary conditions between spatially distant free tiles.

	\paragraph{Initialising the TM}
	Finally, boundary conditions are needed to force the correct initial configuration of the Turing Machine.
	To ensure this, \cite{Robinson_1971} introduces a further set of tile markings that interact with the Turing Machine markings.
	The markings are chosen so that every arm tile which is both horizontally facing and forms the bottom border of an $n$-border, and which \emph{does not} absorb an obstruction signal, must emit a Turing Machine signal upwards.
	Choose this signal to be $s_0$ which will force the tiles in the initial layer at free positions to be blank, so that the initial tape configuration is entirely blank.

	The exception to this is in the centre of the edge where the tile will emit a Turing Machine signal $s_0q_0$ indicating the Turing Machine head starts there.
	Similarly, choose the tiling markings so any arms in the top, left and right parts of the square's border will absorb any stray Turning Machine signals along their inner edges.

	\section{Robinson Tiling Robustness}\label{Sec:Robinson_Robustness}

	In this section we prove a series of results demonstrating that if a region $R\subseteq \mathbb{Z}^2$ is tiled with Robinson tiles, but tiling defects are allowed to occur (i.e.\ points between adjacent tiles at which the matching rules are not satisfied), then only a finite number of Robinson squares can be destroyed per defect.
	Similarly bounds were proven in \cite{Miekisz_1997, Cubitt_PG_Wolf_Undecidability}, but are not strong enough for our purposes.

	\subsection{Robinson Border Deficit Bound}

\input{Robinson-robustness.tex}

	\subsection{Obstruction Signal Bound}
	We now bound the deficit of squares which have correct obstruction tilings.
	As discussed in \cref{Sec:Encode_TM}, obstruction signals are used to demarcate free tiles.
	We now add these markings to the modified Robinson tiles: we make a small change relative to the obstruction markings in \cite{Robinson_1971} and choose all obstruction signals to have a direction: the horizontal  left-to-right or downwards.
	This new set of tiles modified Robinson tiles + obstruction markings are labelled \emph{obstruction tiles}.
	\begin{definition}[Correct Obstruction Tiling] \label{Def:Obstruction_Tiling}
		A complete Robinson square has a \emph{correct obstruction tiling} if:
		\begin{enumerate}
			\item a tile has no obstruction signals iff it is in both a free row and free column. \label{Obstruction_1}
			\item a tile has horizontal (vertical) obstruction signals run across it iff it is in a free column (row). \label{Obstruction_2}
			\item all tiles not contained a free row or column of the $n$-border, and not contained within another $m$-border, have obstruction signals running across them  both horizontally and vertically. \label{Obstruction_3}
		\end{enumerate}
	\end{definition}

Consider the deficit in the number of squares present with correct obstruction tilings in the case of defects versus defect-free tilings:
\begin{lemma}\label{Lemma:obstruction_deficit}
	Let $T$ be a tile configuration of a finite subregion of $\Z_2$ with perimeter of length $L$,
	and $D$ its defect set.

	Define the \emph{obstruction deficit} of $T$, $\odeficit(T)$, to be the difference between the total number of complete Robinson squares in $T$ with a correct internal obstruction tiling, and the number of these in a Robinson tiling of the same region, maximised over Robinson tilings.

	The obstruction deficit of $T$ is bounded by
	\begin{equation}
	\odeficit(T) \leq 800\abs{D}+2L.
	\end{equation}
\end{lemma}
\begin{proof}
\Cref{def:square_deficit} and \cref{square_deficit}, bound the square deficit.
Given any complete square, obstruction signals which are emitted by the smaller interior borders can terminate if (a)~there is a defect in their path, or (b)~they end on a tile with double red arrows.
Case~(a) implies there is a defect in the interior of the square.
Case~(b) implies there a tile with a double red arrow marking which is horizontal (vertical) in a free column (row), which must immediately result in a defect if contained in a square.
Conversely, a tile without obstruction signals which is not placed in a free row must be adjacent to a tile with obstruction signals, and thus must cause a defect.
It follows that if any of the conditions from \cref{Def:Obstruction_Tiling} are not met, then there must be an interior defect.
Furthermore, if there is a complete $n$-square with a defect, the tiling outside the $n$-border is unaffected as its $n$-border is identical to the case without a defect.
Thus
\begin{align}
\odeficit(T) \leq  \sdeficit(T) +  |D|
\end{align}

\end{proof}

	\subsection{Robinson + TM Tiling Bound}

	Finally, the obstruction tiles need to be combined with the obstructions tiles, which are themselves a combination of the Robinson and obstruction tiles.
	We use the full set of tiles described in \cref{Sec:Encode_TM} which include the Robinson markings, obstruction markings, and Turing Machine signals.
	\begin{definition}
		An $n$-square has a correct TM encoding if its $(2^n+1)\times (2^n+1)$ free tiles encode the correct evolution of a TM from some fixed initial state according to the TM's transition rules.
	\end{definition}

	\begin{lemma} \label{Corollary:Total_Deficit}
		Let $T$ be a tile configuration of a finite subregion of $\Z_2$ with perimeter of length $L$, $D$ its defect set.

		Define the \emph{total deficit} of $T$, $\tdeficit(T)$, to be the difference between the total number of complete Robinson squares in $T$ with a correct internal Turing Machine tiling, and the number of these in a Robinson tiling of the same region, maximised over Robinson tilings.

		The total deficit of $T$ is bounded by
		\begin{equation}
		\tdeficit(T) \leq 801\abs{D}+2L.
		\end{equation}
	\end{lemma}
	\begin{proof}
		All tiles with no obstruction markings present must have TM markings (but not TM signal markings) and visa-versa.
		Thus, assuming an $n$-square has correct obstruction markings, TM markings only appear on free tiles, and the TM signals only appear on tiles with obstruction markings going horizontally \textbf{or} vertically, but not both.

		All tiles with only a horizontal (vertical) obstruction markings have a TM signal running vertically (horizontally).
		Since by \cref{Def:Obstruction_Tiling}, such tiles only appear in the appropriate free column (row), the TM signals only run along the free columns (rows).
		The TM signals propagate until they reach a free square, at which point they may change.
		If a TM signal changes between tiles, not mediated by a free tile, there must be a defect.
		Thus
		\begin{align}
		\tdeficit(T) \leq \odeficit(T) + |D|.
		\end{align}

	\end{proof}

	\section{Proofs of GSED Complexity Results}\label{sec:proofs}

	\subsection{Classical hardness for \PNEEXP}
	In this section we set out to prove the following theorem:
	\begin{theorem}
		$\PNEEXP{}\subseteq\EXPGSED{}$, for \GSED{} as defined in \cref{Def:Decision_Problem_GSED_2}, for a classical, nearest-neighbour, translationally invariant Hamiltonian.
	\end{theorem}

	\noindent To prove this result, we will show that it is possible to encode the outputs of a doubly-exponential time nondeterministic TM in the ground state energy density of a particular, fixed, classical Hamiltonian.

	~\newline
	\noindent
	\textbf{Specifying the Encoded TMs}

	We want to enumerate over all input strings for a TM deciding some language, encode these using tiles, and arrange for the TMs running on different inputs to be encoded within Robinson borders of different sizes.
	This is summed up as:
	\begin{lemma}[TMs in Robinson Squares] \label{Lemma:NEXP+Robinson_Tiling}
		Let $x_n\in \{0,1\}^*$ be the $(n-n_0)^{th}$ string in lexicographic order where $n_0$ is a fixed integer, and let $M$ be a non-deterministic TM.
		It is possible to construct a tile set such that all valid tilings of an
		$L \times L$ lattice consist of the pattern of nested squares formed by the Robinson tiling, such that within each complete $n$-border, $\forall n\geq n_0$, the tiles encode a valid computational evolution of $M(x_n)$ for time $2^{2^{c|x_n|}}$, $c\geq 1$.
	\end{lemma}
	\begin{proof}
		As per \cref{Lemma:UTM_on_Robinson_Square}, we are able to encode a TM in the $(2^n+1)\times (2^n+1)$ grid of free tiles of a Robinson $n$-squares.
		Section 3 of \cite{Gottesman-Irani} proves that given a $L\times L$ grid with an appropriate border, it is possible to encode a computation of length $kL$ and space $L$, for $k=O(1)$.
		Here the Robinson $n$-borders provide such a border.

		We choose to encode a series of TMs as follows.
		This first TM is binary counter machine $M_{BC}$ which after time step $T$, has $T$ written in binary on the tape (see \cite{Patitz2014} or \cite[Section~3]{Gottesman-Irani} for an explicit construction of this machine).
		This outputs the square size $2^n+1$ in binary.
		Then run a TM computing $\log_4(y-1)-n_0$ on this output, which outputs $x\in \{0,1\}^*$, the $(n-n_0)^{th}$ string in lexicographic order.
		Finally encode a non-deterministic TM which takes input $x$ and runs for $2^{2^{(k-2)|x|}}$($\leq 2^{(k-2)n}$) steps.
		We can force $M$ to run for $2^{2^{(k-2)|x|}}$ steps by employing a counter to limit the number of steps to $2^{2^{(k-2)|x|}}$; if the TM halts before reaching end of the allotted time, the final time step is copied to the next time step.
		If the timer runs out before the full grid space is used, the final time step of the encoded TM is copied forwards until the grid is filled.
		Choose $n_0$ to be the smallest integer such that these TMs have enough space to operate properly on a grid of size $(2^{n_0}+1) \times (2^{n_0}+1)$.
	\end{proof}

	Note that, at this point, the tiling here can encode any computational path (even those which reject when there is an accepting path) of the nondeterministic TM $M$ as we have not constrained the output in any way.

	\subsubsection{Mapping Tiles to Hamiltonians}

	So far we have presented the problem in terms of a tiling problem and need to map this to a classical Local Hamiltonian problem.
	This is a standard technique (see \cite[Section~3]{Gottesman-Irani} or the appendix of \cite{Bausch2015} for a summary).
	Consider a set of Wang tiles $\mathcal{T}$ rules with horizontal constraints $\mathcal{R}_{Horz} \subseteq \mathcal{T} \times \mathcal{T}$ such that if $t_i$
	is placed to the left of $t_j$, then it must be the case that $(t_i,t_j)\in \mathcal{R}_{Horz}$ and likewise for the vertical tiling rules $\mathcal{R}_{Vert}$.

	Map every tile type $t_i\in \mathcal{T}$ to spin state of a classical particle $\ket{t_i}$.
	We then impose a Hamiltonian over the lattice such that if the tiling pair $(t_i,t_j)\not \in \mathcal{R}_{Horz}$ (or $(t_i,t_j)\not \in \mathcal{R}_{Vert}$ depending on the orientation), then we introduce the term $\ket{t_it_j}\bra{t_it_j}$ for all forbidden pairings $(t_i,t_j)$ over all points in the lattice.

	Thus we end up with a Hamiltonian composed of local interactions of the form
	\begin{align} \label{Eq:TILING_Hamiltonian}
	h_{k,k+1}^{col} &= \sum_{(t_i,t_j)\not\in \mathcal{R}_{Horz}}\ket{t_it_j}\bra{t_it_j}_{k,k+1}\\
	h_{k,k+1}^{row} &= \sum_{(t_i,t_j)\not\in \mathcal{R}_{Vert}}\ket{t_it_j}\bra{t_it_j}_{k,k+1},
	\end{align}

	We now map the tiling rules produced by \cref{Lemma:NEXP+Robinson_Tiling} to a Hamiltonian to get a nearest-neighbour, translationally invariant Hamiltonian.
	We add a term penalising rejecting instances of the verification computation; $\Pi_{NO}$ is an additional term we add in which assigns an energy penalty to \NO{} problem instances.

	We encapsulate the definition of the Hamiltonian in the following:
	\begin{definition}[Robinson + Computation Hamiltonian]\hfill\\
          \label{Def:Tiling_Hamiltonian}
	Let $h^{col,Rob}, h^{row, Rob} \in \mathcal{B}(\mathbb{C}^R \otimes \mathbb{C}^R)$ be the local terms which encode the local matching rules for the Robinson tiling, obstruction rules and TM rules from \cref{Lemma:NEXP+Robinson_Tiling}.
	Let $(\Pi_{NO})_{j,j+1}$ be a projector onto the reject state of the encoded TM, $M$, on a site in row $j$, and a Robinson border tile on the adjacent site in row $j+1$.
	Then the overall local terms are:
	\begin{align} \label{Eq:Hamilonian_Row}
	h^{row}_{i,i+1} &= \Lambda h^{row, Rob}_{i,i+1} \\
	h^{col}_{j,j+1} &= \Lambda h^{col, Rob}_{j,j+1} + (\Pi_{NO})_{j,j+1}\label{Eq:Hamilonian_Column}
	\end{align}
	where $\Lambda\in\N$ is a parameter that we will fix later.
	\end{definition}
	$\Pi_{NO}$ is constructed such that the energy penalty is only applied at the edge of a Robinson border where a TM has halted in the \NO{} state (i.e. once the TM has stopped running).
	$\Lambda$ characterises the energy penalty for breaking the Robinson tiling, the obstruction signals, or the TM signals.
	We will need to choose $\Lambda$ to be a sufficiently large constant to make it energetically unfavourable to break the Robinson tiling in the ground state.

	\begin{lemma} \label{Lemma:Robinson_Square_Energy}
		Define $H(4^n)|_{P}$ to be the Hamiltonian on a $(4^n-1)\times (4^n-1)$ region described by the local terms given in \cref{Eq:Hamilonian_Row,Eq:Hamilonian_Column},
		restricted to the subspace $P$ corresponding to defect-free tilings of the region that contain a complete Robinson $n$-border.
		Let $x\in \{0,1\}^*$ be the $(n-n_0)^{th}$ string in lexicographic order and let $M$ be a non-deterministic Turing Machine running for time $2^{2^{cm}}$ on inputs of length $m$, $c\geq 1$.

                Then for $n\geq n_0$, the ground state energy of $H(4^n)|_P$ is
		\begin{align}
		\lambda_0(H(4^n)|_{P}) = i_n :=
		\begin{cases}
		0 &  {} M(x){} \text{outputs \YES } \\
		1 & {} M(x){} \text{outputs \NO}.
		\end{cases}
		\end{align}
	\end{lemma}
	\begin{proof}
		$H(4^n)|_{P}$ is restricted to the subspace of valid tiling configurations containing a complete Robinson $n$-border. Clearly, this border must run around the edge of the $(4^n-1)\times (4^n-1)$ region. By \cref{Lemma:NEXP+Robinson_Tiling} valid tilings encode the evolution of a non-deterministic TM $M(x)$, where $x$ is the $(n-n_0)^{th}$ string in lexicographic order.
		By restricting to the subspace $P$ we have ensured the encoded TM evolves correctly.

		If $x$ is a \YES{} instance, then $M(x)$ must have an accepting computational path, and so there must be a set of states that encode the correct evolution which finishes in an accepting state.
		Hence there is no energy penalty and the ground state is $0$.

		If $x$ is a \NO{} instance, then there is no accepting path.
		Any correct evolution of $M(x)$ therefore enters the rejecting state, and the tile marking the rejecting state of the TM picks up an energy penalty of $1$ from the term $(\Pi_{NO})_{k, k+1}$ (and no other state receives this energy penalty).

	\end{proof}

	\subsubsection{Robustness of the Ground State}
	We now want to find the ground state energy of the lattice with Hamiltonian from \cref{Def:Tiling_Hamiltonian}.
	The possible energy contributions come from tiling defects and energy penalties for \NO{} instances of the encoded computation.
	In the following, we use the square deficit bounds established in \cref{Sec:Robinson_Robustness} to show that it is energetically unfavourable to have too many tiling defects, regardless of how many \NO{} instances might be encoded in $n$-squares.

	\begin{lemma}[Robinson Square Bound]\label{Lemma:defect-free}
          The  number of $n$\nobreakdash-borders in a Robinson tiling of $\Lambda(L\times H)\subset \Z^2$ using modified Robinson tiles is bounded by $\ge (\lfloor H/2^{n+1}\rfloor-1) \bigl(\lfloor L/2^{n+1}\rfloor -1\bigr)$ and $\le \bigl(\lfloor H/2^{n+1}\rfloor +1 \bigr)(\lfloor L/2^{n+1}\rfloor+1)$ for all $n$.
	\end{lemma}
	\begin{proof}
		A Robinson border is completely contained in an $L\times H$ lattice iff its top edge and its left edge are completely contained in the lattice.
		Lemma~48 of \cite{Cubitt_PG_Wolf_Undecidability} shows that the number of top edges of a Robinson $n$-square which are completely contained in the $L\times H$ lattice is $\ge \lfloor H/2^{n+1}\rfloor \bigl(\lfloor L/2^{n+1}\rfloor -1\bigr)$ and $\le \bigl(\lfloor H/2^{n+1}\rfloor +1 \bigr)\lfloor L/2^{n+1}\rfloor$.
		From this it is straightforward to see the number of left edges which are completely contained in the lattice is $\ge \bigl(\lfloor H/2^{n+1}\rfloor-1\bigr) \lfloor L/2^{n+1}\rfloor $ and $\le\lfloor H/2^{n+1}\rfloor  \bigl( \lfloor L/2^{n+1}\rfloor +1 \bigr)$.

		Combining these two bounds gives $\ge (\lfloor H/2^{n+1}\rfloor -1) \bigl(\lfloor L/2^{n+1}\rfloor -1\bigr)$ and $\le \bigl(\lfloor H/2^{n+1}\rfloor +1 \bigr)(\lfloor L/2^{n+1} \rfloor+1)$.
	\end{proof}

	We now want to check that the ground state of the Hamiltonian on the overall lattice is a tiling of the lattice with Robinson squares in which a verification TM is encoded as we expect, but potentially with a bounded number of defects.

	\begin{lemma} \label{Lemma:Energy_Bounds}
		Let $h^{row}, h^{col} \in \mathcal{B}(\mathbb{C}^{R}\otimes \mathbb{C}^{R})$ be the local interactions that encode the tiling rules given by \cref{Eq:Hamilonian_Row,Eq:Hamilonian_Column}.
		Let $H^{\Lambda(L \times L)}$ be the Hamiltonian with these local interactions on $\Lambda(L\times L)$.

                Then for sufficiently large $L$, the ground state energy $\lambda_0(H^{\Lambda(L \times L)})$ is contained in the interval
		\begin{align}
		\bigg[& \sum_{n=n_0}^{\lfloor \log_4(L/2) \rfloor} \bigg(\bigg\lfloor\frac{L}{2^{2n+1}}\bigg\rfloor -1\bigg)^2\lambda_0(H(4^n)|_{P}) +\Lambda|D| - k_1|D|-k_2L, \nonumber \\ &\sum_{n=n_0}^{\lfloor \log_4(L/2) \rfloor}\bigg(\bigg\lfloor\frac{L}{2^{2n+1}}\bigg\rfloor +1\bigg)^2\lambda_0(H(4^n)|_{P})+\Lambda|D| - k_1|D|-k_2L   \bigg]
		\end{align}
		for some constants $\Lambda$, $k_1$ and $k_2$ such that $\Lambda \gg k_1 + k_2$, and $|D|=O(L)$.
        \end{lemma}
	\begin{proof}

		From \cref{Lemma:Robinson_Square_Energy}, we see that in the ground state energy contribution from each sufficiently large, complete, Robinson $n$-square is $\lambda_0(H(4^n)|_{P}) \in \{0,1\}$.
		By \cref{Lemma:defect-free}, the number of $n$-borders of a given size in an $L\times L$ region with no defects is bounded by $\geq(\lfloor L/2^{2n+1} \rfloor -1)^2  $ and $ \leq (\lfloor L/2^{2n+1} \rfloor +1)^2$.

		Let $N(D)$ denote the number of borders correctly encoding the TM evolution for some tile configuration $T$ with defect set $D$. Let $N_{YES}(D)$, $N_{NO}(D)$ be the number of borders which encode \YES{} and \NO{} instances, respectively. Hence $N(D) = N_{YES}(D)+N_{NO}(D)$.
		Let $E(|D|\  \text{defects})$ be the energy of a configuration with $|D|$ defects.
		Then
		\begin{align}
		E(|D|\  \text{defects}) &= \Lambda|D| + N_{NO}(D)\\
		 E(0\  \text{defects}) &= N_{NO}(\emptyset)
		\end{align}
		Combining these:
		\begin{align}
		E(|D|\  \text{defects}) - E(0\  \text{defects}) &= \Lambda|D| -(N_{NO}(\emptyset) - N_{NO}(D)) \\
		E(|D|\  \text{defects}) - E(0\  \text{defects}) &\geq  \Lambda|D| -(N(\emptyset) - N(D))
		\end{align}
		where the fact $N(\emptyset) - N(D)\geq N_{NO}(\emptyset) - N_{NO}(D)$ has been used.

		\cref{Corollary:Total_Deficit} gives $\tdeficit(T) = N(\emptyset) - N(D)  \leq k_1|D|+k_2L$ for constants $k_1,k_2$, hence
		\begin{align}
			E(|D|\  \text{defects}) - E(0\  \text{defects}) &\geq  \Lambda|D| -(k_1|D|+k_2L)
		\end{align}

		Now choose the parameter $\Lambda$ to be constant such that $\Lambda\gg k_1+k_2$.
		If $|D|=\Omega(L)$, then for sufficiently large $L$,
		\begin{align*}
		E(|D|{} \text{defects}) - E(0{} \text{defects}) &\geq (\Lambda  -k_1-k_2)\Omega(L) k
		= \Omega(L).
		\end{align*}
		Thus, for sufficiently large $L$, the 0-defect case becomes the ground state.

		If $|D|=O(L^{1-o(1)})$, then for sufficiently large $L$ we have that
		\begin{align*}
		E(|D|{} \text{defects}) - E(0{} \text{defects}) &\geq \Lambda |D|-k_1|D|-k_2L
		= -O(L).
		\end{align*}
		Thus we see  the minimum lower bound occurs for $|D|=O(L^{1-o(1)})$

		There is one energy contribution that has been omitted. Some Robinson squares will be too small to have the TM's encoded in them run correctly.
		However, there are only finitely many square sizes for which this is the case, and each square size appears with constant density. So their contribution to to the ground state energy density is a constant which can be computed in constant time, and subtracted off with a 1-local term of the form $\sum_{i\in \Lambda(L\times L)}\alpha \1_i$. (Cf.~\cite{Cubitt_PG_Wolf_Undecidability}.)
              \end{proof}

                For simplicity of the exposition, we omit the above constant energy shift from the expressions and discussion, as it does not affect the analysis.

	\begin{lemma} \label{Lemma:Ground_State_Energy_Density}

		Consider an $L \times L$ lattice with a local Hamiltonian interactions given by \cref{Eq:Hamilonian_Row,Eq:Hamilonian_Column}, and let $H(4^n)|_{P}$ and $i_n$ be defined as in \cref{Lemma:Robinson_Square_Energy}.
		In the limit of $L \rightarrow \infty$, the ground state energy density is
		\begin{equation}\label{Eq:GSED_Expression}
		\Er = \frac{1}{4} \sum_{n=n_0}^{\infty} \frac{\lambda_0(H(4^n)|_{P})}{16^n} = \frac{1}{4} \sum_{n=n_0}^{\infty} \frac{i_n}{16^n}.
		\end{equation}
	\end{lemma}
	\begin{proof}
		By \cref{Lemma:Energy_Bounds}, we have bounds on the ground state energy for the region:
		\begin{align}
                  \begin{split}
		&\sum_{n=n_0}^{\lfloor \log_4(L/2) \rfloor} \frac{1}{L^2}  \bigg(\bigg\lfloor\frac{L}{2^{2n+1}}\bigg\rfloor -1\bigg)^2 \lambda_0(H(4^n)|_{P})+ (\Lambda- k_1)O(L^{-1})+k_2L^{-1} \\
           &\leq  \Er(H^{\Lambda(L \times L)}) \\
		&\leq \sum_{n=n_0}^{\lfloor \log_4(L/2) \rfloor} \frac{1}{L^2}\bigg(\bigg\lfloor\frac{L}{2^{2n+1}}\bigg\rfloor +1\bigg)^2\lambda_0(H(4^n)|_{P}) )+ (\Lambda- k_1)O(L^{-1})+k_2L^{-1} \raisetag{5em}
                \end{split}
		\end{align}
		Taking the limit $L\rightarrow \infty$ gives
		\begin{align}
		\lim\limits_{L \rightarrow \infty}\Er(H^{\Lambda(L \times L)}) =\Er = \frac{1}{4} \sum_{n=n_0}^{\infty} \frac{\lambda_0(H(4^n)|_{P})}{16^n}
		= \frac{1}{4} \sum_{n=n_0}^{\infty} \frac{i_n}{16^n}.
		\end{align}
        \end{proof}

	We now prove the main theorem, which we restate here for convenience.
 	\begin{theorem}[\PNEEXP$\subseteq$\EXPGSED]\label{Theorem:PNEEXP_Containment}
		\PNEEXP$\subseteq$\EXPGSED, for \GSED{} as defined in \cref{Def:Decision_Problem_GSED_2}, for a classical, translationally invariant, nearest-neighbour Hamiltonian.
	\end{theorem}
	\begin{proof}
		Consider any polytime bounded TM $M_1$. 
		We will show we can simulate $M_1^{\NEEXP}$ with $M_2^{\GSED}$ where $M_2$ is another exptime TM.
		If $M_1^{ \NEEXP}$ takes an $n$-bit input, it can then make $O(\poly(n))$ queries. Denote these queries by $\{q_i\}_{i=1}^{O(\poly(n))}$.
		Each individual query must have length $|q_i| = O(\poly(n))$.
		The $M_1$ machine then runs for an $O(\poly(n))$ time and produces some output.

		To simulate this, $M_2$ takes the $n$-bit input and calculates each of the queries which $M_1$ makes: $\{q_i\}_{i=1}^{O(\poly(n))}$.
		Each query $q_i$ is made to a \NEEXP{} oracle. 
		So $M_2$ takes each query $q_i$, and reduces it to an instance of determining the output of a doubly-exponentially time non-deterministic TM, $M$, on input $y_i$. 
		This reduction can be computed in polynomial time, as the problem of determining the output of double-exponential-time non-deterministic TMs is manifestly \NEEXP{}-hard. (Note by using padding arguments we can reduce any language in $\NEEXP$ to $\textsf{NTIME}(2^{2^{cn}})$ for some $c>1$ \cite{Papadimitriou_1994}).
		This defines a new set of inputs to the non-deterministic machine $M$, $\{y_i\}_{i=1}^{O(\poly(n))}$, such that $|y_i|=O(\poly(n))$.

		Now order the $\{y_i\}_i$ lexicographically and take the largest one. Suppose the largest string, $y_j$, is the $k^{th}$ string in lexicographic order. Then $k =O(2^{O(|y_j|)})=O(2^{\poly(n)})$.

		We will use the \GSED{} oracle for the Hamiltonian of \cref{Def:Tiling_Hamiltonian} to perform a binary search in order to obtain a sufficiently precise approximation to the ground state energy density $\Er$, such that we can extract the result of computing $M$ on all inputs up to $y_j$. 
		To do this, we need to query the $\GSED${} oracle on all the instances before it in lexicographic order, of which there are $k=O(2^{\poly(n)})$ many.

		By \cref{Lemma:Ground_State_Energy_Density}, outputs $i_n$ to the queries $\{y_i\}_i$ are encoded as
		\begin{equation}
		\Er = \frac{1}{4}\sum_{n=n_0}^{\infty}\frac{i_n}{16^n}.
		\end{equation}
		We extract the $i_k$ iteratively as follows. 
		Assume for simplicity that $n_0=1$. 
		(If this is not the case, $n$ can trivially be adjusted appropriately.)
		To determine the $i_1$, note that if $i_1=0$, then the maximum $\Er$ can be is
		\begin{equation}
		\Er = \frac{1}{4}\sum_{n=2}^{\infty} \frac{1}{16^n} = \frac{1}{960}
		\end{equation}
		and otherwise the minimum it can be is $1/64$.
		Hence we ask whether $\Er\geq \beta_1=1/64$ or $\Er\leq \alpha_1 = 1/960$.
		Thus
		\[i_1=
		\begin{cases}
		0 & \text{if } \Er<1/960 \\
		1 & \text{if } \Er>1/64.
		\end{cases}
		\]

		We can then perform a similar process for all $i_m$, $1\leq m<k$, assuming we have previously extracted $i_1,i_2,\dots ,i_{m-1}$.
		When extracting the $m^{th}$ instance, we have that either $\Er \leq \alpha_m$ or $\Er \geq \beta_m$, where
		\begin{align}
		\beta_m &= \frac{1}{4}\bigg( \frac{1}{16^m} + \sum_{n=1}^{m-1}\frac{i_n}{16^{n}}   \bigg)\nonumber \\
		\alpha_m &= \frac{1}{4} \bigg( \sum_{n=1}^{m-1} \frac{i_n}{16^n} +   \sum_{n=m+1}^{\infty} \frac{1}{16^n}  \bigg). \label{Eq:alphabeta_m}
		\end{align}

		Since $y_j$ is the $k^{th}$ string in lexicographic order, $k=O(2^{\poly(n)})$, the maximum precision we need to go to is $\Omega(2^{-2^{\poly(n)}})$, which is possible provided $\alpha_m, \beta_m$ can have binary length $|\alpha_m|, |\beta_m|=O(2^{\poly(n)})$.
		Since $M_2$ is an exponential time machine, it has time and space to write these strings to the oracle tape.
		Furthermore, $M_2$ only needs to make $O(2^{\poly(n)})$ queries.
		Thus $M_2^{\GSED}$ is able extract all the answers to the queries made by $M_1^{\NEEXP}$, and hence after making these queries and performing the relevant post-processing, output the solution.
	\end{proof}


	\subsection{Classical Containment in \EXPNEXP}
	We now need to show that for classical \GSED, as defined in \cref{Def:Decision_Problem_GSED_2},  \EXPGSED $\subseteq$ \EXPNEXP.
	The first step is to show that the ground state energy density of a finite $L\times L$ part of the lattice is a good estimate for the energy density of the full lattice~\cite{Cubitt_PG_Wolf_Undecidability}:
	\begin{lemma}\label{Lemma:E_rho_Approximation}
	Consider a translationally invariant, nearest-neighbour Hamiltonian on $\Lambda(L\times L)$ lattice defined by local terms $h^{row}_{i,i+1}, h^{col}_{j,j+1}$.
	Let $\Er(L)$ be the energy density of the Hamiltonian on this lattice, and
	let $\Er$ be the energy density in the $L\rightarrow \infty$ limit.
	Then
		\begin{align}
			|\Er(L)-\Er|=\frac{4 \max \left\{\norm{h^{row}_{i,i+1}},\norm{h^{col}_{i,i+1}}\right\}  }{L}.
		\end{align}
	\end{lemma}
	\begin{proof}
	Let $H(L)$ be the Hamiltonian defined on $\Lambda(L\times L)$ and let $t\in \N$.
	Let $H_{grid}(L,t)$ be the Hamiltonian with the same local terms, but with the terms $h_{i,i+1}^{row}, h_{j,j+1}^{col}$ removed for $i,j\in t\mathbb{N}$.
	Then:
	\begin{align}\label{Eq:H_grid}
		H_{grid}(L,t)=H(tL) - \sum_{i\mod t=0} h^{row}_{i,i+1} - \sum_{j\mod t=0} h^{row}_{j,j+1}.
	\end{align}
	The interaction graph of $H_{grid}(L,t)$ is a set of $t^2$ squares of size $L\times L$.
	Hence equation \ref{Eq:H_grid} gives
	\begin{align*}
		\norm{H_{grid}(L,t)-H(tL)}\leq  4t^2L\max\left\{\norm{h^{row}_{i,i+1}},\norm{h^{col}_{i,i+1}}\right\}.
	\end{align*}
	It is straightforward to see that $\lambda_0(H_{grid}(L,t)) =t^2 \lambda_0(H(L))$.
	Combining these gives
	\begin{align*}
		|t^2\lambda_0(H(L)) - \lambda_0(H(tL))| \leq 4Lt^2 \max\left\{\norm{h^{row}_{i,i+1}},\norm{h^{col}_{i,i+1}}\right\}.
	\end{align*}
	Dividing through by $t^2L^2$ to get energy densities gives
	\begin{align}
	|\Er(L) - \Er| \leq \frac{4\max\left\{\norm{h^{row}_{i,i+1}},\norm{h^{col}_{i,i+1}}\right\}}{L}.
	\end{align}

	\end{proof}

	\begin{lemma}\label{Lemma:GS_in_NEXP}
	\GSED{} $\in$ \NEXP{} for any classical, nearest-neighbour, translationally invariant Hamiltonian, for \GSED{} as defined in \cref{Def:Decision_Problem_GSED_2}.
	\end{lemma}

	\begin{proof}
	$(\alpha, \beta)$ is the input of the problem, $\beta- \alpha = \Omega(2^{-q(n)})$.
	We show an \EXP{} machine will be able calculate $\Er(L)$ (using the notation of \cref{Lemma:E_rho_Approximation}) using a classical witness for $L=2^{p(n)}$, for a polynomial $p$.

	First compute the ground state energy of an $L\times L$ square of the lattice. Take as the witness the ground state of the Hamiltonian restricted to an $L\times L$ region of the lattice: $\ket{\psi}=\ket{\phi_1}\otimes\ket{\phi_2}\otimes \dots \ket{\phi_{L^2}}$, where $\ket{\phi_i}\in \mathbb{C}^{\abs{\mathcal{S}}}$ is the state of the spin at lattice site $i$.
	Now,
	\begin{align}
	\Er(L) = \frac{1}{L^2}\sum_{\langle i,j \rangle}\bra{\phi_i}\bra{\phi_j}h_{i,j}\ket{\phi_i}\ket{\phi_j}, \nonumber
	\end{align}
	where $\langle i,j \rangle$ denotes pairs of nearest-neighbours.
	$\bra{\phi_i}\bra{\phi_j}h_{i,j}\ket{\phi_i}\ket{\phi_j}$ can be computed in $O(1)$ time, and there are $O(L^2)$ such terms.
	Since $L=2^{p(n)}$, the estimate $\Er(L)$ can be computed in $O(L^2)=O(2^{2p(n)})$ time.
	By \cref{Lemma:E_rho_Approximation}, $|\Er(L)-\Er|=O(L^{-1})$, hence provided we choose $p(n)$ to be sufficiently large relative to $q(n)$, the approximation $\Er(L)$ allows us to determine $\Er > \beta$ or $\Er<\alpha$ for $\beta-\alpha=\Omega(2^{-q(n)})$.
	\end{proof}

	\begin{lemma}[\EXPNEXP{} Containment] \label{Lemma:NEXP-Completeness}
		\EXPGSED $\subseteq$ \EXPNEXP, for \GSED{} as defined in \cref{Def:Decision_Problem_GSED_2}, for a fixed, classical Hamiltonian.
	\end{lemma}
	\begin{proof}
		For \EXPGSED$\subseteq $\EXPNEXP{} we show that, given an exponential time TM $M_1$ with access to an oracle \GSED{}, its action can be simulated by an exponential time TM $M_2$ with oracle access to \NEXP.

		Consider the action of $M_1^{\GSED}$. If it takes an $n$-bit input, it may make $O(\exp(n))$ queries, each of length $O(\exp(n))$, before outputting an answer based on these query outcomes.
		Each query must be in the form of an $(\alpha, \beta)$ such that $\beta-\alpha = \Omega (2^{-2^{p(n)}})$ for some polynomial $p$.

                The $(\alpha, \beta)$ queries made by $M_1$ must have input length of $|q_i|=O(\exp(n))$
		By \cref{Lemma:GS_in_NEXP} determining whether $\Er > \beta$ or $\Er < \alpha$ for $\beta - \alpha = \Omega(2^{-|q_i|}) = \Omega(2^{-2^{-g(n)}})$ is contained in \NEXP.
                Thus $M_2^{\NEXP}$ can simulate the queries to \GSED{} by making querying the \NEXP{} oracle, and hence the entire action of $\EXP^{\GSED}$.
		\end{proof}

	\paragraph{Why not polytime Turing Reductions, $\mathsf{P}^{\GSED}$?}

	Naturally a question arises as to why we consider $\EXPGSED$ here, rather than $\mathsf{P}^{\GSED}$.
	Here we show that using our hardness construction, one cannot even hope to prove $\NP \subseteq\mathsf{P}^{\GSED}$ unless the polynomial hierarchy collapses to $\Sigma^P_2$.

	\begin{lemma}\label{Lemma:GSEDh_Containment}
	Let $\mathsf{P}^{\GSED_h}$ be the class of languages decided by a polynomial time oracle machine with access to a $\GSED$ oracle for the Hamiltonian of \cref{Def:Tiling_Hamiltonian} only.
	Let $\mathsf{P}_{\log}^O$ be the languages decided by a polytime oracle machine with oracle $O$ which is only able to make $\log(n)$ length queries to the oracle for an $n$-bit input.
	Then  $\mathsf{P}^{\GSED_h} \subset \mathsf{P}^{\NEEXP}_{\log}$.
	\end{lemma}
	\begin{proof}
	Let $M_1^{\GSED_h}$ be a polytime TM with oracle access to a $\GSED$ oracle for the Hamiltonian defined in \cref{Def:Tiling_Hamiltonian} only.
	Let $M_2^{\NEEXP}$ an oracle machine which can only make $O(\log(n))$ length queries to the oracle.
	We will show the latter machine can simulate the former.

	$M_1^{\GSED_h}$ can make at most $O(\poly(n))$ length queries to the oracle, corresponding to $\alpha, \beta$ queries such that $\beta - \alpha =\Omega(2^{-p(n)})$ for some polynomial $p$.
	After making at most $\poly(n)$ queries, it performs some post-processing and finally outputs an answer.

	$M_2^{\NEEXP}$ can simulate this by simply calculating $\Er$ for the Hamiltonian in \cref{Def:Tiling_Hamiltonian} by querying the \NEEXP \ oracle for the first $O(\log(n))$ instances, and then computing an estimate for $\Er$, denoted $\tilde{\Er}$, using equation \cref{Eq:GSED_Expression}.
	By making sufficiently many queries to the \NEEXP \ oracle, one can make it so $|\tilde{\Er}- \Er|=O(2^{-q(n)})$ for some polynomial $q$.
	Thus by making $q(n)\gg p(n)$, $M_2$ can then simulate all the queries that $M_1^{\GSED_h}$ makes, do the same post-processing, and output the same answer.
\end{proof}

\begin{theorem}
	Using the notation defined in \cref{Lemma:GSEDh_Containment},
	if $\NP \subseteq\mathsf{P}^{\GSED_h}$, then the polynomial hierarchy collapses to $\Sigma^P_2$.
\end{theorem}
\begin{proof}
	From \cref{Lemma:GSEDh_Containment}, $\mathsf{P}^{\GSED_h} \subseteq \mathsf{P}^{\NEEXP}_{\log}$.
	Now note $\mathsf{P}^{\NEEXP}_{\log}\subseteq P/poly$. 
	This is true because for an input of length $n$, $\mathsf{P}^{\NEEXP}_{\log}$ can make at most $O(\poly(n))$ different queries. 
	Hence we could simply give a TM a $O(\poly(n))$ length advice string giving the answers to each of these queries, such that the advice string only depends on the input length $n$.

	Thus $\mathsf{P}^{\GSED_h}\subseteq \mathsf{P}^{\NEEXP}_{\log}\subseteq P/poly$.
	However, it is known that if $\NP \subseteq P/poly$, then the polynomial hierarchy collapses to $\Sigma^P_2$ \cite{Karp_Lipton_1980}.

\end{proof}
This provides strong evidence that our hardness construction is not $\NP$-hard under polytime Turing reductions.

	\subsubsection{Improving the Hardness Result}\label{Sec:PSPACE_Results}

	We can improve our containment and hardness results by using a \PSPACE{}  oracle machine.
	There is, however, some controversy as to how a \PSPACE{}  oracle machine should have access to its oracle; in particular whether the input tape to the oracle has a polynomial space bound or not \cite{Buss_1988, Hartmanis_Chang_Kadin_Mitchell_1993, Fortnow_94}.
	Here we consider both of these definitions and show how they can be used to tighten our complexity bounds on \GSED.

	\begin{definition}[$1^{st}$ \PSPACE \ Oracle Machine Definition]\hfill\\
		A $\PSPACE^O${} \ oracle machine is a $\PSPACE${} \ machine with access to an oracle input tape, for which it can make \textbf{polynomial length} queries to the oracle.
	\end{definition}

	\noindent For this definition we get:
	\begin{theorem}
		$\PSPACE^{\NEEXP}\subseteq \EXPGSED$.
	\end{theorem}
\begin{proof}
	Identical to the proof for \cref{Theorem:PNEEXP_Containment} except $M_1$ is now a $\PSPACE$ machine which needs to be simulated by the \EXPGSED \ oracle machine.
	\end{proof}

A potentially more interesting result occurs when we use the following definition:

	\begin{definition}[$2^{nd}$ \PSPACE \ Oracle Machine Definition]\hfill\\
	A $\PSPACE^O$ \ oracle machine is a $\PSPACE${} \ machine with access to a write only oracle input tape, for which it can make \textbf{exponential length} queries to the oracle.

\end{definition}
This is the preferred definition of several authors \cite{Ladner_Lynch_1976,Fortnow_94}.
For this definition of oracle machine, we realise that one can do the binary search protocol used in the proof of \cref{Theorem:PNEEXP_Containment} to get:
\begin{theorem}
	$\PNEEXP\subseteq \PSPACE^{\GSED}$.
\end{theorem}
\begin{proof}
	The proof will be similar to the proof for \cref{Theorem:PNEEXP_Containment}, except now the \PSPACE{} machine will have to make exponentially long oracle calls to the \GSED{} oracle for to extract the query results while using only polynomial space everywhere else.

	Let $M^{\GSED}$ be a $\PSPACE$ machine with (for convenience) two work tapes\footnote{This can be reduced to a single work tape by standard arguments.} (bounded by polynomial space) and one unbounded oracle tape which is read only. Let the $\GSED$ oracle be the one for the Hamiltonian of \cref{Def:Tiling_Hamiltonian}.
	Let $M^{\GSED}$ have made $(k-1)$ queries to the oracle machine with outputs $i_1, i_2 \dots i_{k-1}$, for $i_j$ as defined in \cref{Lemma:Robinson_Square_Energy}, such that it now needs to make a $k^{th}$ query.
	To do so, it needs to calculate a pair $(\alpha_k, \beta_k)$ which will allow it to extract $i_k$.
	Assume $M$ has the string $i_1 i_2 \dots i_{k-1}$ stored on one of the two work tapes. We need to write out the numbers $\alpha_k,\beta_k$ in binary as given in equation \cref{Eq:alphabeta_m}.

	Without loss of generality, assume the oracle input tape is initially in the all $0$ state.
	To write out $\beta_k$ on the input tape, $M$ take a query outcome $i_j$, then moves $4j+2$ down the tape and places $i_j$ in the $(4j+2)^{th}$ cell (corresponding to value $\frac{1}{4}\frac{i_j}{16^j}$).
	Finally in the $(4k+2)^{th}$ cell it places a~1.
	To determine where the head is on the oracle input tape, we let $M$ have a binary counter on its second work tape.
	$M$ can determine where the head is on the input tape by increment/decrementing the binary counter whenever the head moves right/left.

	$M$ cannot write out $\alpha_k$ exactly, as it does not have a finite binary expansion.
	Instead, upper bound it by a number $a_k> \alpha_k$, $\beta-a_k=\Omega(2^{-\poly(k)})$ which does have a finite expansion
	\begin{align}
		a_k =\frac{1}{4}\left( \sum_{n=1}^{k-1} \frac{i_n}{16^n} + \frac{2}{16^{k+1}}\right)  > \alpha_k.
	\end{align}
	To write this out, $M$ also places $i_j$ in the $(4j+2)^{th}$ cell, for $j\leq k-1$.
	We then place a 1 in the $(4k+3)^{th}$ cell (which is the contribution from the $2\times 16^{-k-1}$ term).
	Hence querying the oracle for $(a_k,\beta_k)$ gives the same answer as querying with $(\alpha_k,\beta_k)$.

	$M$ then continues with the computation until all the necessary queries have been extracted.
	Since only $\poly(n)$ many queries are made, the \PSPACE{} machine is capable of storing them all on its work tape.
	It can then post-process the queries and output the answer to the relevant \PNEEXP{} computation.

        Since $M$ only needs to record the number of queries $k=O(\poly(n))$ and the binary counter it uses to keep track of the TM head on the input string --- which counts up to $16^{O(\poly(k))}$ --- can be expressed in $\poly(k)=\poly(n)$ bits, we have that $M$ only uses $\poly(n)$ space on its two work tapes, as required.
      \end{proof}

This result maybe should not be too surprising given that it is known how to do binary search procedures using exponentially less space.
For example, if $\mathsf{L}$ is a logspace machine: $\mathsf{P}^{\mathsf{SAT}}=\mathsf{L}^{\mathsf{SAT}}=\mathsf{L}^{\mathsf{SAT}[\log(n)]}=\mathsf{L}^{||\mathsf{SAT}}$ \cite{Wagner_1988}

	The results from this section immediately give:
	\begin{corollary}
		$\PNEEXP\subseteq \PSPACE^{\GSED}\subseteq \PSPACE^{\NEXP}$
	\end{corollary}

        \subsubsection{Complexity Results for \FGSED}
We show containment of the function problem version \FGSED{} of the ground state energy density problem:
\begin{theorem}
	$\FGSED\in\FP^{\GSED}\subseteq \FP^{\NEXP}$ for classical \FGSED.
\end{theorem}
\begin{proof}
	Let $\epsilon$ be the input to \FGSED, such that $|\epsilon|=n$.
	Let $M^{\GSED}$  be a polytime TM with oracle access to \GSED.
	Then using $\poly(n)$ many $(\alpha,\beta)$ queries to \GSED, for $\beta - \alpha = \Omega(2^{-\poly(n)})$, we can use a binary search procedure to find an estimate $\tilde{\Er}$ such that $|\tilde{\Er} - \Er|=O(2^{-\poly(n)})<\epsilon$.
	Thus a $M^{\GSED}$ machine can compute $\FGSED$.
	Since $\GSED\in \NEXP$, this implies $\FGSED\in\FP^{\GSED}\subseteq \FP^{\NEXP}$.
\end{proof}

\begin{lemma}
	$\FP^{\NEEXP} \subseteq \FEXP^{\FGSED}\subseteq \FEXP^{\NEXP}$ for classical \FGSED.
\end{lemma}
\begin{proof}
	To show $ \FEXP^{\FGSED}\subseteq \FEXP^{\NEXP}$, consider two exponential time oracle machines $M_1^{\FGSED}$ and $M_2^{\NEXP}$.
	Let $M_1$ make $O(\exp(n))$ oracle calls to $\FGSED$, and then do some exponential time post-processing.
	$M_2$ can simulate these oracle calls by, for each oracle call $M_1$ makes, estimating using the \NEXP \ oracle $\exp(n)$ to estimate the ground state energy density produced by $\FGSED$.
	Since $M_1$ makes $\exp(n)$ queries, $M_2$ needs to make $O(\exp(n))\times O(\exp(n))=O(\exp(n))$ queries.
	It can then perform the same post-processing as $M_1$.
	Thus   $ \FEXP^{\FGSED}\subseteq \FEXP^{\NEXP}$.

	To show $\FP^{\NEEXP} \subseteq \FEXP^{\FGSED}$, consider a polytime oracle machine $M_3^{\NEEXP}$ and an exptime oracle machine $M_4^{\FGSED}$.
	$M_3$ can make at most $O(\poly(n))$ queries to the \NEEXP{} oracle of at most $O(\poly(n))$ length, and then do post-processing to output the relevant function.
	$M_4$ can simulate all of these queries by asking the \FGSED \ oracle for an estimate for $\epsilon$ such that $|\epsilon|=O(\exp(n))$, from which it can extract all the $\NEEXP$ queries.
	It can then do the relevant post-processing and output the same function as $M_3$.

\end{proof}

	\subsection{Quantum Containment in \EXPQMAEXP}
		In this section we show containment of \GSED{} for quantum Hamiltonians.
		\begin{lemma}\label{Lemma:GS_in_QMAEXP}
			\GSED{} $\in$ \QMAEXP{} for any quantum, nearest-neighbour, translationally invariant Hamiltonian, for \GSED{} as defined in \cref{Def:Decision_Problem_GSED_2}.
		\end{lemma}
		\begin{proof}
		$(\alpha, \beta)$ is the input of the problem for $\beta- \alpha = \Omega(2^{-p(n)})$.
		Let $\ket{\psi}$ be the ground state an $L\times L$ section of the lattice, for $L=2^{q(n)}$, which our \QMAEXP \ machine will take as a witness.
		Perform quantum phase estimation of $e^{iH^{\Lambda(L)}}$ to $q(n)$ bits of precision, which gives an estimate $\tilde{\lambda_0}$ of $\lambda_0(H^{\Lambda(L)})$ such that $|\tilde{\lambda_0}-\lambda_0(H^{\Lambda(L)})|\leq 2^{-p(n)}$, and takes time $O(2^{q(n)})$ \cite{Nielsen_and_Chuang}.

		Since $\Er(L)= \tilde{\lambda}_0$, and by \cref{Lemma:E_rho_Approximation} that $|\Er(L) - \Er|=O(2^{-p(n)})$,
		choosing $q(n)$ to be sufficiently larger than $p(n)$ allows us to verify whether $\Er>\beta$ or $\Er < \alpha$.
		\end{proof}

		\begin{corollary}
			$\EXPGSED \subseteq  \EXPQMAEXP$ for a fixed, nearest-neighbour, translationally invariant \textbf{quantum} Hamiltonian.
		\end{corollary}
		\begin{proof}
			The proof is identical to \cref{Lemma:NEXP-Completeness}, but making use of \cref{Lemma:GS_in_QMAEXP}.
		\end{proof}
	Since classical Hamiltonians are a subset of quantum Hamiltonians, the following result is an immediate corollary of \cref{Theorem:PNEEXP_Containment}:
	\begin{corollary}
		\PNEEXP$\subseteq$\EXPGSED{} for a fixed, nearest-neighbour, translationally invariant \textbf{quantum} Hamiltonian.
	\end{corollary}

		\section{Discussion and Conclusions}\label{sec:conclusions}

		\paragraph{Quantum GSED}
		A natural question to ask is if tighter results can be found for \GSED{} for quantum Hamiltonians.
                As we have seen, it follows straightforwardly that \EXPGSED$\subseteq $\EXPQMAEXP, but a non-trivial quantum lower bound does not follow easily.

		Our proof of a \PNEEXP{} lower bound works as we can enumerate over \NEEXP-complete problems.
		Attempting to prove a similar quantum lower bound (e.g. \PQMAEXP) runs into the problem that, since \QMAEXP{} is a promise class, for a given \QMAEXP-complete problem there may be instances which do not satisfy the promise (so called ``invalid queries'').
		This makes it impossible to enumerate over all instances of a given \QMAEXP-complete problem without potentially including instances which do not satisfy the promise.
		It is not currently known how to avoid these instances from occurring, although some techniques exist, such as \cite{Gharibian_Yirka_2019, Gharibian_Piddock_Yirka_2020, Watson_Bausch_Gharibian_2020}.

		\paragraph{Closing the Classical Upper and Lower Bounds}
		So far we have separate lower and upper bounds \PNEEXP and \EXPNEXP.
		The containment protocol given here works via a natural binary search algorithm to determine $\Er$, and as such we believe it is optimal.
		While it is not immediately clear how the lower bound might be improved, it is not clear whether the construction presented here should give a tight lower bound.

		\paragraph{Other Precision Problems}
		As far as the authors know, this is the first complexity result about a theorem in which the only input parameter which is varied is the precision, but where the object of study is fixed.
		Furthermore, \GSED{} can be viewed as a precision version of the Local Hamiltonian problem; can similar  ``precision based'' problems be developed for other decision/promise problems?
		Is there a natural situation in which they occur?
		
		 ~\newline
		\noindent\textbf{Note Added} \\
		Whilst preparing this paper for submission, we became aware of parallel work by Irani \& Aharonov on similar topics. 
		See their paper in the same arXiv listing as this one .


	\section{Acknowledgements}
	The authors would like to recognise useful discussions with  J. Lockhart, S. Gharibian, and J. Sullivan.
	J.D.W. is supported by
	the EPSRC Centre for Doctoral Training in Delivering Quantum Technologies
	(grant EP/L015242/1).
	T.S.C. is supported by the Royal Society.
	This work has been supported in part by the EPSRC
	Prosperity Partnership in Quantum Software for Simulation and Modelling
	(grant EP/S005021/1), and by the UK Hub in Quantum Computing and
	Simulation, part of the UK National Quantum Technologies Programme with
	funding from UKRI EPSRC (grant EP/T001062/1).

	\printbibliography

\end{document}

%% file: intro.tex
The connection between computational complexity theory and many-body physics dates back over 40 years.
Barahona's \cite{Barahona1982} proof of \NP-completeness of the ground state energy problem for classical many-body models with local interactions\footnote{Namely, the 2D Ising model with fields.} --- or ``local Hamiltonians'' for short --- on a finite number of particles (spins), established the ground state energy as one of the canonical physical quantities for which computational complexity yields insight.

The \emph{Hamiltonian} is the function mapping a state of the particles to its corresponding energy.
The ground state is then the minimum energy state of the system, and the ground state energy that minimum energy value.
The problem of estimating the ground state energy is often formulated as an equivalent (up to polynomial-time computation) decision problem known as the Local Hamiltonian problem: given a Hamiltonian and an energy threshold, decide whether the ground state energy is above or below that threshold.

Nearly 20 years later, Kitaev \cite{Kitaev2002} proved \QMA-completeness (the quantum analogue of \NP-completeness) for quantum local Hamiltonians on a finite number of quantum particles.
There has been a plethora of papers following --- too many to comprehensively list here --- building on Barahona and Kitaev's seminal results.
These have extended hardness of the ground state energy problem to ever more restrictive classes of Hamiltonian, with specific, physically-motivated types of local interaction, and with restricted patterns of local interaction.
In particular,
amongst many other related results, we now know that the classical and quantum ground state energy problems remain \NP- and \QMA-complete when restricted to nearest-neighbour interactions on a finite 2D square lattice and a finite 1D chain, respectively~\cite{Barahona1982,Aharonov_Gottesman_Irani_Kempe_2007}.
Properties beyond the ground state energy have been studied, including density of states \cite{Brown_Flammia_Schuch_2011}, expectation values on low energy subspaces \cite{Ambainis_2014}, the energy of excited states \cite{Jordan_Gosset_Love_2010}, detecting energy barriers \cite{Gharibian_Sikora_2018}, determining whether a system is frustrated, and many others.

The input to all of the above problems is a description of a local Hamiltonian on a finite number of particles, and the complexity-theoretic hardness is a function of varying the Hamiltonian.
However, many-body and condensed matter physicists are more often interested in properties of a many-body system in the \emph{thermodynamic limit} of infinitely many particles.
Most many-body physics properties, such as phase transitions, phase diagrams, spectral gaps, etc., are only well-defined theoretically in this limit.
Moreover, in experimental physics, these models often arise as idealisations of physical materials, where a typical sample will contain such a large number of atoms that the properties of the material are well-approximated by the infinite limit.

Furthermore, they are typically interested in computing the physical properties of a \emph{single} Hamiltonian -- or a family of Hamiltonians parametrised by a small, constant number of parameters.
Often, the local interactions have some regular structure, such as translational invariance where all the local interactions take the same form.
The standard formulation of the ground state energy problem does not capture this type of question.

\subsection{Related work}
There are a small number of results proving hardness of estimating the ground state energy for a translationally invariant Hamiltonian where the local interaction is fixed, and the only input to the problem is the lattice size.
Here, since a lattice of size $2^n$ can be specified in $n$ bits, the natural complexity class is \NEXP{} (or \QMAEXP \ in the quantum case), rather than \NP.
The Wang tiling completion problem is known to be \NEXP{}-complete~\cite{Papadimitriou,Gottesman-Irani}, which can trivially be translated to the ground state energy problem for a single, fixed, translationally invariant, nearest-neighbour, classical Hamiltonian on a 2D square lattice, where the state at some of the boundaries is fixed (fixed boundary conditions).
As the interaction is fixed, the only remaining problem input is the size of the lattice.
Remarkably, this alone suffices for the hardness result.
Gottesman and Irani~\cite{Gottesman-Irani} also extended these results to more natural types of boundary condition.
They went on to prove the analogous \QMAEXP-completeness result for quantum Hamiltonians on a 1D chain.
However, these results still concern Hamiltonians on finite numbers of particles; indeed, the problem input is the number of particles the Hamiltonian acts on.

In the thermodynamic limit, the ground state energy is no longer a meaningful quantity; it typically has infinite magnitude, and is not physically measurable.
In this setting, the more relevant quantity is the ground state energy \emph{density}: the minimum energy \emph{per particle}.
Just as the ground state energy is a key starting point for studying the physics of finite many-body systems, the ground state energy density (GSED) is a key starting point for physics in the thermodynamic limit.
Methods of approximating the ground state energy density in condensed matter systems have been the subject of much study in the physics literature~\cite{Perdew_et_al_1992,Hollenberg_Witte_1994}.

Less is known about the computational complexity of the ground state energy density problem, than for the ground state energy.
Gottesman and Irani~\cite{Gottesman-Irani} proved that the ground state energy density problem for translationally invariant, nearest-neighbour, quantum Hamiltonians on a 1D chain with a $\Omega(1/2^n)$ promise-gap is \NEXP{}-complete.
Here, the input is a description of the local interaction of the system, and the complexity is a function of varying over the Hamiltonian.
Meanwhile, as a stepping stone to their undecidability result for the spectral gap, \cite{Cubitt_PG_Wolf_Nature,Cubitt_PG_Wolf_Undecidability} proved that deciding whether the ground state energy density is 0 or strictly positive, with no promise gap, is undecidable,
Their result holds for quantum, translationally invariant, nearest neighbour Hamiltonians on a 2D square lattice with a fixed local dimension.
\cite{Bausch_1D_Undecidable} later extended this undecidability result to 1D chains (again as a stepping stone to the spectral gap problem) and \cite{Bausch_Cubitt_Watson_2019} extends to to 2D systems for which the local interaction are analytic in the input parameter.

However, as with most ground state energy complexity results, these results still have as input the description of the Hamiltonian, and the hardness is a result of varying the Hamiltonian.

\subsection{The Ground State Energy Density problem}
If we restrict to a single, fixed Hamiltonian in the thermodynamic limit, it may seem that there are no input parameters left, and complexity theory can have nothing to say!
However, this is not quite the case.
We can still ask about the complexity of estimating the ground state energy density to a given precision, where the only input is the precision required.
(See \cref{sec:main_results} for precise problem definitions.)
Arguably, this is the problem formulation closest to that often encountered in condensed matter physics.

If we learn the ground state energy density to precision $2^{-n}$, then we can hope to learn the first $n$ bits of its binary representation.
An $n$ bit string can encode the solutions to at most $n$ different decision problems.
But an index into this bit string, specifying the index of the decision problem we are interested in, requires only $\log n$ bits.
Therefore, the natural complexity class for \GSED{} is \NEEXP, or related doubly-exponential time complexity classes.
(At least for hardness results.)

In this work, drawing on techniques developed in \cite{Cubitt_PG_Wolf_Undecidability}, we prove upper and lower bounds on the complexity of this Ground State Energy Density (\GSED) problem: we show that \GSED{} is \NEEXP-hard under exponential-time Turing reductions, and contained in \EXPNEXP.
In fact, we prove the following slightly stronger results for the natural promise-problem formulation of \GSED, for a fixed, classical, translationally invariant, nearest-neighbour Hamiltonian on a 2D square lattice:
\begin{equation}
  \PNEEXP \subseteq \EXPGSED \subseteq \EXPNEXP
\end{equation}
The natural promise-problem formulation of \GSED{} takes as input two energy density thresholds $\alpha$ and $\beta$ with $\beta-\alpha = \Omega(2^{-n})$, and
outputs whether the ground state energy density is above $\beta$ or below $\alpha$.

The analogous complexity bounds for the function problem formulation of \GSED{} readily follow from this.
For the function problem formulation, the input is the precision $\epsilon$, and the output is an estimate of the ground state energy density to precision $\epsilon$.

For quantum Hamiltonians, a very similar argument to the classical case establishes the analogous upper bound of \EXPQMAEXP \ for the quantum \GSED{} problem.
The same lower bound as above follows trivially from the fact that classical Hamiltonians are a special case of quantum.
However, we are not able to prove \QMAEEXP-hardness of the quantum problem.
(We comment on this briefly in \cref{sec:conclusions}.)

The ground state energy density of the specific Hamiltonian we construct is a single, real number $\Er$.
Our hardness results imply the solutions to \emph{all} instances of \NEEXP{}-complete problem are encoded in the digits of this single number, with successive digits of $\Er$ giving the solution to successive instances of a canonical \NEEXP-complete problem.
In this sense, the ground state energy density of this Hamiltonian is somewhat reminiscent of Chaitin's constant~\cite{Chaitin1975}, but encoding solutions to problems in a certain complexity class, rather than the Halting problem.

\subsection{Proof techniques}
We draw on the construction and proofs in \cite{Cubitt_PG_Wolf_Undecidability}, used to prove undecidability of the spectral gap for quantum Hamiltonians.
However, in order to obtain our \GSED{} results, we apply those techniques in a quite different way.

\cite{Cubitt_PG_Wolf_Undecidability} showed how to encode an arbitrary Quantum Turing Machine into a translationally invariant \emph{quantum} Hamiltonian overlayed on an aperiodic tiling due to Robinson~\cite{Robinson_1971}, in such a way that the ground state energy density is related to the output of the computation.
However, we instead follow Robinson's original construction~\cite{Robinson_1971} to encode a classical Turing Machine (TM) into a classical Hamiltonian, overlayed on the same aperiodic tiling.
Robinson's aperiodic tiling forms a nested pattern on squares of all possible sizes of the form $4^n$ (see \cref{sec:tiling_preliminaries}.
He then encodes a TM in each square, such that the ground state encodes infinitely many copies of the same TM running on all possible finite tape lengths of the form $4^n$, with the density of each tape length falling off as $1/16^n$.

\cite{Robinson_1971} (and \cite{Cubitt_PG_Wolf_Undecidability}) use this construction to encode a universal TM in order to prove undecidability results.
Here, we instead use this to encode a \NEEXP{} machine, such that the ground state picks up an additional $O(1)$ energy contribution if the computation rejects.
Furthermore, we use an idea from \cite{Papadimitriou,Gottesman-Irani} to first run a binary counter TM which computes the length of the tape it is running on in binary, and feeds this as input to the \NEEXP{}  machine.
In this way, instead of all copies of the TM carrying out the same computation (albeit with different bounds on the length of tape available), the copies \NEEXP{}  machine are each computing different problem instances.
Specifically, the copies running on tape length $4^n$ are computing problem instance $n$, for all possible values of $n$.

Since the density of TMs for each $n$ falls off as $1/16^n$, the energy from this instance rejecting contributes to the $4n$'th digit of the ground state energy density $\Er$, when  this is expressed as a binary fraction.
The solution to the $n$'th \NEEXP{} problem can therefore be extracted from $\Er$ by binary search.

An more extensive overview of the necessary background is given in \cref{sec:preliminaries}.
The rigorous proof following the above argument is given in \cref{sec:proofs}.

\subsection{Robinson robustness}
However, for this proof to go through, one must show that the ground state --- and hence its energy density --- is indeed of the desired form.
This is non-trivial, as configurations that break the aperiodic tiling can prevent the encoded TMs from ``running'', thereby avoiding the energy contribution from the computation rejecting.
On the infinite lattice, this could potentially reduce the energy by an unbounded amount.

In fact, proving that the \emph{quantum} \GSED{} problem is \NEEXP-hard is significantly more straightforward (given the results of \cite{Cubitt_PG_Wolf_Undecidability}) than the stronger classical hardness result we prove in his paper (which trivially implies hardness of the quantum case).
To prove the quantum result directly, we can follow exactly the same construction as we do here for the classical result (see \cref{sec:proofs}), but using the construction of \cite{Cubitt_PG_Wolf_Undecidability} to encode the \NEEXP{}  machine into a quantum Hamiltonian, rather than a classical one.
In the quantum case, the Robinson tiling ``rigidity'' results already proven in \cite[Section 5]{Cubitt_PG_Wolf_Undecidability} then suffice to show that the ground state has the required form, and the argument goes through.
However, the previously known Robinson rigidity results are too weak to prove this in the classical case.

The reason the quantum case is easier to prove is that the energy contribution from rejecting computations itself falls off as $1/4^n$.
Thus the sum of the energy contributions over all infinitely many values of $n$ is still $O(1)$.
This makes it easier to prove that the reduction in energy from avoiding these contributions, is outweighed by the energy penalty from breaking the aperiodic tiling.
It suffices to prove that breaking the aperiodic tiling pattern at one site can destroy at most one square of \emph{of each size} in aperiodic pattern.
Thus it can only prevent one TM for each value of $n$ from ``running''.
This could still affect infinitely many of the TMs.
But as long as it only affects one of each size, the total energy contribution is summable and bounded by $O(1)$.
(See \cite[Section 5]{Cubitt_PG_Wolf_Undecidability}.)

However, in the classical case, the energy contribution from rejecting computations is $O(1)$, independent of $n$.
The sum of the energy contributions over all $n$ is therefore infinite.
Thus breaking the aperiodic tiling could potentially result in a unbounded reduction in the ground state energy.
The results of \cite[Section 5]{Cubitt_PG_Wolf_Undecidability} are too weak to rule this out.
Instead, we must prove a Robinson rigidity result that is, in some sense, infinitely stronger than what was previously known.
We must prove that breaking the tiling pattern at one site can destroy at most $O(1)$ squares \emph{in total} in the tiling pattern.

Proving this stronger bound requires more sophisticated techniques than the corresponding bound in \cite[Section 5]{Cubitt_PG_Wolf_Undecidability}.
We use an intricate combination of combinatorial and geometric arguments to relate the number of destroyed squares to Delaunay triangulations of defects appearing in the aperiodic tiling pattern.
The properties of Delaunay triangulations then allow us to prove the requisite strong bound.

This new robustness result for Robinson tilings may be of independent interest, and a self-contained proof is given in \cref{Sec:Robinson_Robustness}.


%% file: Robinson-robustness.tex
\begin{lemma}\label{no_border_overlap}
  In any tile configuration $T$, borders cannot overlap.
\end{lemma}

\begin{proof}
  This follows immediately from \cref{def:n-border} of $n$-borders: if two borders were to overlap, the lattice cell where they overlap would necessarily contain the wrong tile for one or other (or both) of the putative borders.
\end{proof}

Throughout this section, all lengths and distances are with respect to the $\ell_\infty$ metric.

\subsubsection{Domains and Undomains}
\begin{definition}[Defect set, defect graph]
  A \keyword{defect set} $D$ is a finite set of points on the dual lattice $\Z_2^*$.
  A \keyword{defect graph} $G=(D,E)$ is the complete graph on $D$ embedded as a line graph in $\R_2$, with vertices at all defects in $D$ and edges $E$ formed by straight lines between all pairs of vertices.
\end{definition}

\begin{definition}[Tile configuration]
  A \keyword{tile configuration} is an assignment of a Robinson tile to each point in the lattice $\Z_2$.
  The defect set of a tile configuration $T$ is the set of all points in $Z_2^*$ between non-matching tiles in $T$.
\end{definition}

\begin{definition}[Tiling]
  A \keyword{Robinson tiling} (or just \keyword{tiling}) is a defect-free tile configuration.
\end{definition}

\begin{definition}[Border intersection]
  We say that a border \keyword{intersects a defect} if the border contains two points in $\Z_2$ that are either side of a point in the defect set.

  We say that a border \keyword{intersects an edge} if the edge passes through or along the side of a lattice cell containing a border tile.
\end{definition}

\begin{definition}[$n$-domain, $n$-undomain]\label{def:n-domain}
  Let $D$ be a defect set, $G=(D,E$) its defect graph.

  We define an \keyword{$n$-undomain} $\cU\subset\Z_2$ to be a maximal connected region of the lattice such that any $m$-border with $m\geq n$ that overlaps $\cU$ necessarily either intersects a defect in $D$, or intersects an edge in $E$ of length $\leq 4^n$.

  We define an \keyword{$n$-domain} to be a maximal connected region $\cD\subset\Z_2$ of the lattice that does not overlap any $n$-undomain.

\end{definition}

We define a defect to be contained in a domain if it is adjacent to any point in $\Z_2$ contained in the domain (i.e.\ considered as regions of the lattice, domains are \emph{closed} -- they contain their boundaries).
In contrast, we define a defect to be contained in an undomain if it is strictly contained in the \emph{interior} of the undomain (i.e.\ undomains are \emph{open} -- they do not include their boundaries).

The following property of $n$-domains and undomains is immediate from the definitions.
\begin{lemma}\label{n-domain_partition}
  The set of all $n$-domains and $n$-undomains partition the region being tiled, with $n$-domains separated by $n$-undomains and vice versa.
\end{lemma}

We will need to establish some further basic properties.

\begin{lemma}\label{undomain_edge_distance}
  A lattice cell contained in an $n$-undomain has an edge of length $\leq 4^n$ within distance $\leq 4^n$.
\end{lemma}

\begin{proof}
  Consider any $n$-border running through the lattice cell.
  Recall that each side of an $n$-border is $4^n-1$ cells long, so the entire $n$-border is within distance $\leq 4^n$ of the lattice cell in question.
  Since the latter is contained in an $n$-undomain, by \cref{def:n-domain} the $n$-border must intersect a defect or an edge of length $\leq 4^n$, which is therefore within distance $\leq 4^n$ of the lattice cell.

  If it intersect an edge, we are done.
  Thus it remains to consider the case in which \emph{every} $n$-border running through the lattice cell intersects a defect.
  But for this to be the case, at least two of those defects must necessarily be within distance $\leq 4^n$ of each other
  thus the edge connecting them fulfils the requirements of the \namecref{undomain_edge_distance}.
\end{proof}

\begin{lemma}\label{domain_containment}
  An $n+1$-domain is contained (not necessarily strictly) within one $n$-domain.
\end{lemma}

\begin{proof}
  It suffices to prove that any lattice cell in the $n+1$-domain is contained in an $n$-domain, since by \cref{def:n-domain} these must then constitute (part of) the same $n$-domain.

  To that end, consider a lattice cell $p$ in the $n+1$-domain.
  By \cref{def:n-domain}, there must exist at least one way that an $n+1$-border can run through $p$ without intersecting any defects or any edges of length $\leq 4^{n+1}$.

  Assume for contradiction that $p$ is \emph{not} contained in an $n$-domain.
  Hence, by \cref{n-domain_partition}, it is contained in an $n$-undomain.
  Thus, by \cref{def:n-domain}, any $n$-border running through $p$ must intersect a defect or an edge of length $\leq 4^n$.
  Consider two such $n$-borders, one within and one outside the $n+1$-border, with a common side running along part of the $n+1$-border.
  If any of the defects or edges intersected by these $n$-borders lie along their common side, then the original $n+1$-border also intersects that defect or edge, contradicting the original condition on the $n+1$-border.
  So the $n$-borders can only intersect defects or edges along sides they do not have in common.
  Thus each of the two $n$-borders must must intersect a \emph{different} defect or edge.

  If the $n$-border in the interior (exterior) of the $n+1$-border intersects a defect along one of the sides not common to both $n$-borders, then there is a defect inside (outside) the $n+1$-border within a distance $\leq 4^n$ of the $n+1$-border.
  Now consider the case that the $n$-border in the interior (exterior) intersects an edge of length $\leq 4^n$ along one of the sides not common to both $r$-borders.
  Since this edge cannot intersect the $n+1$-border, the defects it runs between must be in the interior (exterior) of the $n+1$-domain.
  Thus again we have a defect inside (outside) the $n+1$-border.
  Furthermore, since the edge has length $\leq 4^n$, the defect is at distance $\leq 4^n$ from the $n$-border, thus $\leq 2\cdot 4^n$ from the $n+1$-border.

  Thus, in all cases, we have a pair of defects, one inside and one outside the $n+1$-border, separated by a distance $\leq 2\cdot 2\cdot 4^n = 4^{n+1}$.
  But this implies the edge running between that pair of defects intersects the $n+1$-border and has length $\leq 4^{n+1}$, which is again in contradiction to the original condition on the $n+1$-border.

  Therefore, the only possibility is that $p$ \emph{is} in fact contained in an $n$-domain, as required.
\end{proof}

The following corollary is immediate from \cref{domain_containment}:
\begin{corollary}\label{domain_tree}
  The set of all domains equipped with the set inclusion relation, $(\{\cD^{(n)}_i\}_{n,i}, \subseteq)$, forms a tree, which we refer to as the \keyword{domain tree}.
\end{corollary}

\subsubsection{Border deficit}
\begin{definition}[Border deficit]\label{def:border_deficit}
  The \keyword{$n$-border deficit}, $\deficit_n(T,S,R)$, in a region $S$ of a tile configuration $T$ with respect to Robinson tiling $R$, is the (magnitude of the) difference between the total number of complete $n$-borders of $T$ within $S$, and the number of complete $n$-borders of $R$ within $S$.

  The \keyword{total border deficit}, $\deficit(T)$, of a tile configuration $T$ is the difference between the total number of complete borders in $T$ and the number of complete borders in a Robinson tiling of the same region, maximised over Robinson tilings.
\end{definition}

The following are immediate from the definition:

\begin{lemma}\label{n-deficit}
  Let $T$ be a tile configuration of a region $S$.
  Then
  \begin{equation}
    \deficit(T) = \max_R \sum_n \deficit_n(T,S,R).
  \end{equation}
\end{lemma}

\begin{lemma}\label{deficit_union}
  Let $T$ be a tile configuration, $A$ and $B$ be arbitrary (not necessarily disjoint) subregions.
  Then
  \begin{equation}
    \deficit_n(T,A\cup B,R) \leq \deficit_n(T,A,R) + \deficit(T,B,R).
  \end{equation}
\end{lemma}
Note that this inequality may be strict even when the regions $A$ and $B$ are disjoint and $A\cup B$ is simply connected, since borders in a Robinson tiling of $A\cup B$ that straddle the boundary between $A$ and $B$ may contribute to the deficits on the right hand side, without contributing to the deficit on the left hand side.

We will make use of the following notation.
Let $T$ be a tile configuration, $D$ a defect set, $E$ a set of edges in a defect graph, $R$ a Robinson tiling, and $S\in\Z^2$ a region of the lattice.
$N_D(n,S,R)$ denotes the number of (complete or partial) $n$-borders of $R$ in $S$ that intersect a defect in $D$ when extended to a complete border.
$N_E(n,S,R)$ denotes the number of (complete or partial) $n$-borders of $R$ in $S$ that (when extended to a complete border) intersect an edge in $E$ of length $\leq 4^n$ , but do \emph{not} intersect a defect in $D$.
When we are considering a single defect $d$ or a single edge $e$, we will write $N_d$ (respectively $N_e$) instead of $N_{\{d\}}$ (respectively $N_{\{e\}}$).
$N_B(n,S,R)$ denotes the number of partial $n$-borders of $R$ that intersect the boundary of the region $S$, but do not intersect a defect in $D$.
$N_\partial(n,S,R)$ denotes the number of partial $n$-borders of $R$ in $S$ that intersect the boundary of the entire tile configuration $T$, but do not intersect a defect in $D$.

\begin{lemma}[Undomain border deficit]\label{undomain_count}
  Let $T$ be a tile configuration, $D$ its defect set, and $G=(D,E)$ its defect graph.
  Consider an $n$-undomain $\cU$ of $T$.
  The $n$-border deficit of $\cU$ with respect to Robinson tiling $R$ is bounded by
  \begin{equation}
    \deficit_n(T,\cU,R) \leq N_D(n,\cU,R) + N_E(n,\cU,R).
  \end{equation}
\end{lemma}

\begin{proof}
  By \cref{def:border_deficit} of border deficit, the $n$-border deficit cannot be greater than the number of $n$-borders in $R$.
  (At most you can lose all the borders.)
  By \cref{def:n-domain} of $n$-undomains, any $n$-border in $\cU$ must either intersect a defect, or intersect an edge of length $\leq 4^n$.
  The bound follows.
\end{proof}

\begin{proposition}\label{domain_tiling}
  Let $T$ be a tile configuration, $D$ its defect set. 
  Within any $n$-domain $\cD$, $T$ contains the same periodic pattern of $m$-borders for all $m\leq n$ as a Robinson tiling of $\cD$, except where an $m$-border would intersect a point in $D$.
\end{proposition}

\begin{proof}
  It follows from \cref{def:n-domain} that any points of $D$ contained strictly within the $n$-domain must be separated by distance $>4^n$; if two defects are separated by less than this, the lattice cells through which the edge between them passes form (part of) an $n$-undomain.

  The argument in the proof of \cite[Lemma~47]{Cubitt_PG_Wolf_Undecidability} shows that, for each $m\leq n$, the Robinson tiles force $m$-borders to form exactly as in a Robinson tiling, except where:
  \begin{itemize}
  \item\label[condition]{domain_tiling:free-path}%
    there is no defect-free vertical and/or horizontal path connecting a central cross of an $m$-border to all the other $m$-border central crosses;
  \item%
    an arm meets a defect; or
  \item%
    defects prevent the cross at the centre of a (complete or partial) $m$-border being forced.
  \end{itemize}

  Since defects in the $n$-domain must by definition be separated by $>4^n$, which is greater than the size of an $m$-border for all $m\leq n$, all $m$-border central crosses are connected to each other by vertical or horizontal paths.
  So the first condition is always satisfied within an $n$-domain.

  If an arm meets a defect, then the corresponding Robinson tiling border intersects that defect, matching the condition in the \namecref{domain_tiling}.

  The central cross of an $m$-border is forced unless there is at least one defect between the surrounding $m$-border and its centre (see proof of \cite[Lemma~47]{Cubitt_PG_Wolf_Undecidability}).
  But in that case, the corresponding Robinson tiling $m+1$-border with that central cross at its corner intersects the defect, matching the condition in the \namecref{domain_tiling}.
\end{proof}


\begin{definition}[Robinson-compatible set]\label{def:Robinson_compatible}
  Let $R^{(n)}$ denote a tile configuration of an $n$-domain $\cD^{(n)}$ whose $m$-borders for all $m\leq n$ are in the same location as in a Robinson tiling of $\cD^{(n)}_i$.
  We say that a set of such tilings $\{R^{(n)}_i\}$ is \keyword{Robinson-compatible} if, for all $\cD^{(m)}_j\subseteq\cD^{(n)}_i$, the locations of the $l$-borders in $R^{(m)}_j$ and $R^{(n)}_i$ coincide for all $l\leq m$.
\end{definition}

\begin{corollary}\label{domain_count1}
  Let $T$ be a tile configuration, $D$ its defect set, $G=(D,E)$ its defect graph, and $\{\cD^{(n)}_i\}$ the set of all its domains.
  Let $R$ be a Robinson tiling.
  Then
  \begin{equation}
    \sum_{n,i}\deficit_n(T,\cD^{(n)}_i,R)
    \leq \max_{\{R^{(n)}_i\}} \sum_{n,i}\Bigl( N_D(n,\cD^{(n)}_i,R^{(n)}_i) + N_B(n,\cD^{(n)}_i,R^{(n)}_i) \Bigr),
  \end{equation}
  where the maximisation is over all Robinson-compatible sets $\{R^{(n)}_i\}$.
\end{corollary}

\begin{proof}
  By \cref{domain_tiling}, the number of $n$-borders in an $n$-domain $\cD^{(n)}_i$ is the same as that in some Robinson tiling $R'$ of the same region, except where an $n$-border of $R'$ would intersect a defect or the boundary of $\cD^{(n)}_i$.
  Thus
  \begin{multline}\label{eq:T_borders}
    \text{\# complete $n$-borders of $T$ in $\cD^{(n)}_i$}\\
    \geq \min_{R'}\Bigl(\text{\# complete or partial $n$-borders of $R'$ in $\cD^{(n)}_i$ that do not}\\[-.5em]
                        \text{intersect any $d\in D$ or the boundary of $\cD^{(n)}_i$}\Bigr).\mspace{80mu}
  \end{multline}

  Meanwhile, the number of complete $n$-borders in a given Robinson tiling $R$ of a region is upper bounded by the maximum number of complete \emph{and} partial $n$-borders in any other Robinson tiling $R'$ of the same region:
  \begin{multline}\label{eq:R_borders}
    \forall\text{ Robinson tilings } R,R':\\
    \text{\# complete $n$-borders of $R$ in $\cD^{(n)}_i$}\mspace{150mu}\\
    \leq \text{\# complete and partial $n$-borders of $R'$ in $\cD^{(n)}_i$}.
  \end{multline}

  Thus, by \cref{def:border_deficit}, \cref{eq:T_borders,eq:R_borders} (where all quantities concern only the region $\cD^{(n)}_i$, which we drop from the expressions for brevity):
  \begin{align}
    \deficit_n&(T,\cD^{(n)}_i,R)\notag\\
    &= \Bigl(\text{\# complete $n$-borders of $R$}\Bigr) - \Bigl(\text{\# complete $n$-borders of $T$}\Bigr)\\
    \begin{split}
      &\leq \max_{R'}\Bigl( \text{\# complete and partial $n$-borders of $R'$}\\
      &\mspace{80mu}       - \text{\# complete or partial $n$-borders of $R'$ that do not}\\[-.5em]
      &\mspace{125mu}        \text{intersect any $d\in D$ or the boundary} \Bigr)
    \end{split}\raisetag{3.5em}\\
    \begin{split}
      &\leq \max_{R'}\Bigl(\text{\# complete or partial $n$-borders of $R'$}\\[-.5em]
      &\mspace{90mu} \text{that intersect defects or boundaries in $\cD^{(n)}_i$}\Bigr)
    \end{split}\raisetag{3.3em}\\
    &= \max_{R'} \Bigl( N_D(n,\cD^{(n)}_i,R') + N_B(n,\cD^{(n)}_i,R') \Bigr).
  \end{align}

  Now, any $m$-domain such that $\cD^{(m)}_j\subseteq\cD^{(n)}_i$ contains (a portion of) exactly the same tile configuration $T$ as $\cD^{(n)}_i$.
  So by \cref{domain_tiling} the $m$-border deficit of $\cD^{(m)}_j$ is bounded by the number of $m$-borders of the same Robinson tiling $R'$ that intersect defects or boundaries of $\cD^{(m)}_j$.
  Thus
  \begin{equation}
    \begin{split}
      \deficit_n&(T,\cD^{(n)}_i,R) + \deficit_m(T,\cD^{(m)}_i,R)\\
      &\leq \max_{R'} \Bigl( N_D(n,\cD^{(n)}_i,R') + N_B(n,\cD^{(n)}_i,R') \\
      &\mspace{100mu} +  N_D(m,\cD^{(m)}_j,R') + N_B(m,\cD^{(m)}_j,R') \Bigr).
    \end{split}
  \end{equation}

  Let $R^{(n)}_i$ be a tile configuration of $\cD^{(n)}_i$ that contains $l$-borders for all $l\leq n$ in the same locations as $R'$.
  Since $N_{D/B}(n,\cD^{(n)}_i,R')$ only count $n$-borders, we have
  \begin{equation}
    N_{D/B}(n,\cD^{(n)}_i,R') = N_{D/B}(n,\cD^{(n)}_i,R^{(n)}_i),
  \end{equation}
  and similarly for $R^{(m)}_j$ and $\cD^{(m)}_j$.
  Since $R^{(n)}_i$ and $R^{(m)}_j$ both have $l$-borders in the same locations as $R'$ for all $l\leq n$, they form a Robinson-compatible set by \cref{def:Robinson_compatible}.
  Thus,
  \begin{equation}
    \begin{split}
      \deficit_n&(T,\cD^{(n)}_i,R) + \deficit_m(T,\cD^{(m)}_i,R)\\
      &\leq \max_{\{R^{(n)}_i,R^{(m)}_j\}} \Bigl( N_D(n,\cD^{(n)}_i,R^{(n)}_i) + N_B(n,\cD^{(n)}_i,R^{(n)}_i) \\
      &\mspace{50mu} +  N_D(m,\cD^{(m)}_j,R^{(m)}_j) + N_B(m,\cD^{(m)}_j,R^{(m)}_j) \Bigr),
    \end{split}
  \end{equation}
  where the maximisation is over all Robinson-compatible $R^{(n)}_i$ and $R^{(m)}_j$.

  Applying this to all $\cD^{(n)}_i$ gives the claimed result.
\end{proof}

\begin{lemma}\label{boundary-borders}
  Let $D$ be a defect set, $G=(D,E)$ its defect graph, and $\cD$ an $n$\nbd-domain.
  Any partial $n$-border that intersects the boundary of $\cD$, when extended to a full border, either intersects the boundary of the overall region being tiled, intersects a point in $D$, or intersects an edge in $E$ of length $\leq 4^n$.
\end{lemma}

\begin{proof}
  By \cref{def:n-domain}, $n$-domains are bounded by $n$-undomains.
  Thus any partial border that intersects the boundary of $\cD$ must either meet an $n$-undomain, or the boundary of the region being tiled.
  By \cref{def:n-domain}, any partial $n$-border that intersects an $n$-undomain satisfies one of the claims of the \namecref{boundary-borders}.
\end{proof}

Using \cref{boundary-borders}, we can reformulate \cref{domain_count1}.
\begin{corollary}[Domain border deficit]\label{domain_count}
  Let $T$ be a tile configuration, $D$ its defect set, $G=(D,E)$ its defect graph, and $\{\cD^{(n)}_i\}$ the set of all its domains.
  Let $R$ be a Robinson tiling.
  Then
  \begin{equation}
    \begin{split}
      \max_R\sum_{n,i}&\deficit_n(T,\cD^{(n)}_i,R)\\
      &\leq \max_{\{R^{(n)}_i\}} \sum_{n,i}\Bigl( N_D(n,\cD^{(n)}_i,R^{(n)}_i) + N_E(n,\cD^{(n)}_i,R^{(n)}_i) + N_\partial(n,\cD^{(n)}_i,R^{(n)}_i) \Bigr),
    \end{split}\raisetag{4em}
  \end{equation}
  where the maximisation is over all Robinson-compatible sets $\{R^{(n)}_i\}$.
\end{corollary}

\subsubsection{Weak border deficit bound}
We will need some lemmas to bound these quantities.

\begin{lemma}\label{ND_bound}
  Let $D$ be a defect set in a region $S$, and $R$ a Robinson tiling of that same region.
  Then
  \begin{equation}
    \sum_n N_D(n,S,R) \leq \abs{D}.
  \end{equation}
\end{lemma}

\begin{proof}
  Let $N_d(n)$ denote the number of $n$-borders in $R$ that intersect a specific $d\in D$, so that
  \begin{equation}
    \sum_n N_D(n,S,R) = \sum_n \sum_{d\in D} N_d(n).
  \end{equation}
  Now, borders in Robinson tilings do not overlap, so at most one border in $R$ can intersect $d$.
  Thus
  \begin{equation}
    \sum_n N_d(n) \leq 1.
  \end{equation}
  Putting this together, we have
  \begin{align}
    \sum_n N_D(n,S,R)
    = \sum_n \sum_{d\in D} N_d(n) = \sum_{d\in D} \sum_n N_d(n) \leq \sum_{d\in D} 1
    = \abs{D}.
  \end{align}
\end{proof}

\begin{lemma}\label{Nd_bound}
  Let $T$ be a tile configuration, $\cB$ a branch of its domain tree, $\{R^{(n)}\}$ a Robinson-compatible set for the domains $\cD^{(n)}\in\cB$, and $d$ a defect.
  Then
  \begin{equation}
    \sum_n N_d(n,\cD^{(n)},R^{(n)}) \leq 1.
  \end{equation}
\end{lemma}

\begin{proof}
  Since all $\cD^{(n)}$ are contained in the same branch of the domain tree, they form a totally ordered set under set inclusion, and there is a maximal $\cD^{(m)}\in\cB$ containing all the others, i.e.\ $\forall n: \cD^{(n)}\subseteq \cD^{(m)}$.
  By \cref{def:Robinson_compatible}, the locations of the $l$-borders in $R^{(n)}$ coincide for all $n\geq l$, and are in the same locations as in some Robinson tiling.
  Therefore, there exists a Robinson tiling $R$ of $\cD^{(m)}$ such that, for all $n$, the $n$-borders of $\cD^{(n)}$ are in the same locations as those of $R$.

  Thus
  \begin{equation}
    \sum_n N_d(n,\cD_i^{(n)},R_i^{(n)}) \leq \sum_n N_d(n,\cD^{(m)},R) \leq 1
  \end{equation}
  using \cref{ND_bound} applied to the defect set $\{d\}$ and region $\cD^{(m)}$.
\end{proof}

\begin{lemma}\label{d_branches}
  Let $T$ be a tile configuration, $d$ a defect.
  All domains $\cD\ni d$ are contained in at most two branches of the domain tree.
\end{lemma}

\begin{proof}
  A defect $d\in\Z_2^*$ is adjacent to two lattice sites in $\Z$.
  Since the domains form a tree by set-inclusion, all the domains containing a given lattice site are contained in a single branch of the domain tree.
\end{proof}

\begin{corollary}\label{ND_R_bound}
  Let $T$ be a tile configuration, $D$ its defect set, and $\{R^{(n)}_i\}$ a Robinson-compatible set for its domains $\{\cD^{(n)}_i\}$.
  Then
  \begin{equation}
    \sum_{n,i} N_D(n,\cD^{(n)}_i,R^{(n)}_i) \leq 2\abs{D}.
  \end{equation}
\end{corollary}

\begin{proof}
  Let $\cB$ denote a branch of the domain tree, and $\cB\ni d$ denote a branch containing at least one domain that contains $d$.
  \begin{align}
    \sum_{n,i} N_D(n,\cD^{(n)}_i,R^{(n)}_i)
    &= \sum_{n,i}\sum_{d\in D} N_d(n,\cD^{(n)}_i,R^{(n)}_i)\\
    &= \sum_{d\in D}\sum_{n,i} N_d(n,\cD^{(n)}_i,R^{(n)}_i)\\
    &= \sum_{d\in D}\sum_{\cD^{(n)}_i\ni d} N_d(n,\cD^{(n)}_i,R^{(n)}_i)\\
    &\leq \sum_{d\in D}\sum_{\cB\ni d}\sum_{\cD^{(n)}_i\in\cB} N_d(n,\cD^{(n)}_i,R^{(n)}_i)\\
    &\leq \sum_{d\in D}\sum_{\cB\ni d}\sum_{\cD^{(n)}_i\in\cB} N_d(n,\cD^{(n)}_i,R^{(n)}_i)\\
    &\leq \sum_{d\in D}\sum_{\cB\ni d} 1 \label{eq:Nd_bound}
      \leq \sum_{d\in D} 2 = 2\abs{D},
  \end{align}
  where \cref{eq:Nd_bound} follows from \cref{Nd_bound,d_branches}.
\end{proof}

\begin{lemma}\label{edge-intersection-bound}
  An edge of length $\leq 4^m$ can intersect at most 3 $n$-borders with $n\geq m$ in a Robinson tiling.
\end{lemma}

\begin{proof}
  The $l$-borders in a Robinson tiling repeat periodically separated by distance $4^l$.
  Thus an edge of length $\leq 4^l$ can intersect at most 2 of them.

  $n$-borders with $n>l$ run along gaps between $l$-borders, and at most one such edge runs along each gap.
  Thus the edge can intersect at most one $n$-border with $n>l$.
\end{proof}

\begin{lemma}\label{NE_bound}
  Let $D$ be a defect set in a region $S$, $G=(D,E)$ its defect graph, and $R$ a Robinson tiling of $S$.
  Then
  \begin{equation}
    \sum_n N_E(n,S,R) \leq 3\abs{E}.
  \end{equation}
\end{lemma}

\begin{proof}
  For edges $e\in E$, define
  \begin{equation}
    N_e(n) :=
    \begin{cases}
      \text{number of $n$-borders in $R$ that intersect $e$} & \text{ length } e \leq 4^n\\
      0 & \text{otherwise}
    \end{cases}
  \end{equation}
  so that
  \begin{equation}
    N_E(n,S,R) \leq \sum_{e\in E} N_e(n).
  \end{equation}
  (The inequality is due to the fact that the same $n$-border may be intersected by more than one edge.)

  By \cref{edge-intersection-bound}, $e$ can intersect at most $3$ $n$-borders such that the length of $e$ is $\leq 4^n$.
  Thus
  \begin{equation}
    \sum_n N_e(n) \leq 3.
  \end{equation}
  Putting all this together, we have
  \begin{equation}
    \sum_n N_E(n,S,R) \leq \sum_n\sum_{e\in E} N_e(n)
    = \sum_{e\in E}\sum_n N_e(n) \leq \sum_{e\in E} 3 = 3\abs{E}.
  \end{equation}
\end{proof}

\begin{lemma}\label{Ne_bound}
  Let $T$ be a tile configuration, $\cB$ a branch of its domain tree, $\{R^{(n)}\}$ a Robinson-compatible set for the domains $\cD^{(n)}\in\cB$, and $e$ an edge in the defect graph of $T$.
  Then
  \begin{equation}
    \sum_n N_e(n,\cD^{(n)},R^{(n)}) \leq 3.
  \end{equation}
\end{lemma}

\begin{proof}
  By \cref{def:Robinson_compatible}, the locations of the $l$-borders in $R^{(n)}$ with $l\leq n$ are in the same locations as in a Robinson tiling.
  In particular, for all $n$, the $n$-borders in $R^{(n)}$ located in the same place as some common Robinson tiling $R$.
  Thus $N_e(n,\cD^{(n)},R^{(n)}) = N_e(n,\cD^{(n)},R)$.

  Let $m$ be the smallest integer such that $e$ has length $\leq 4^m$.
  If $m>n$, then $N_e(n,\cD^{(n)},R) = 0$ by definition.
  Thus $\sum_n N_e(n,\cD^{(n)},R) = \sum_{n\geq m} N_e(n,\cD^{(n)},R)$.
  But by \cref{edge-intersection-bound}, an edge $e$ of length $\leq 4^m$ can intersect at most 3 $n$-borders with $n\geq m$ in a Robinson tiling, and the bound follows.
\end{proof}

\begin{definition}\label{def:adjacent_edge}
  We say that an $n$-domain $\cD$ is \keyword{adjacent} to an edge $e\in E$ (or conversely) \keyword{with respect to defect graph $G=(D,E)$} if there exists either a vertical or horizontal path from some lattice cell in $\cD$ to a lattice cell intersected by $e$, such that the path does not intersect any other defect or edge, nor pass through any $n$-domain other than $\cD$, along the way.
\end{definition}

The following is immediate from \cref{def:n-domain,def:adjacent_edge}:

\begin{corollary}\label{adjacent_intersection}
  Let $T$ be a tile configuration, $G=(D,E)$ its defect graph, and $\cD$ a $n$-domain.
  If an $n$-border in $\cD$ (when extended to a full border) is intersected by an edge in $E$ of length $\leq 4^n$, then either it also intersects a defect, or it is also intersected by an edge of length $\leq 4^n$ that is adjacent to $\cD$ with respect to $G$.
\end{corollary}

\begin{lemma}\label{adjacent_edge_dist}
  Let $T$ be a tile configuration, $G=(D,E)$ its defect graph.
  If an edge $e\in E$ of length $\leq 4^n$ is adjacent to an $n$-domain $\cD$, then it must be within distance $\leq 4^n$ of $\cD$.
\end{lemma}

\begin{proof}
  \Cref{n-domain_partition} implies $\cD$ is surrounded by $n$-undomains.
  \Cref{undomain_edge_distance} implies the $n$-undomain cells adjacent to $\cD$ are within distance $\leq 4^n$ of an edge of length $\leq 4^n$.
  The \namecref{adjacent_edge_dist} follows by \cref{def:adjacent_edge}.
\end{proof}

\begin{lemma}\label{e_branches}
  Let $T$ be a tile configuration, $e$ an edge of length $\leq 4^n$ in its defect graph.
  All $m$-domains $\cD^{(m)}$ with $m\geq n$ that are adjacent to $e$, are contained in at most 4 branches of the domain tree.
\end{lemma}

\begin{proof}
  \Cref{adjacent_edge_dist} implies that $n$-domains adjacent to $e$ must be within distance $\leq 4^n$.
  By \cref{def:n-domain}, an $n$-domain has to be at least $4^n$ wide and tall.
  Therefore, at most 10 adjacent $n$-domains fit around $e$.

  By \cref{domain_containment}, an $(m+1)$-domain is contained within an $m$-domain.
  If that $m$-domain is not adjacent to $e$, then by \cref{def:adjacent_edge} there exists no free path from the $m$-domain -- nor hence from the $m+1$-domain -- to $e$.
  Therefore, $m+1$-domains adjacent to $e$ must be contained in $m$-domains adjacent to $e$, and the \namecref{e_branches} follows.
\end{proof}

\begin{corollary}\label{NE_R_bound}
  Let $T$ be a tile configuration, $G=(D,E)$ its defect graph, and $\{R_i^{(n)}\}$ a Robinson-compatible set for its domains $\{\cD_i^{(n)}\}$.
  Then
  \begin{equation}
    \sum_{n,i} N_E(n,\cD^{(n)}_i,R^{(n)}_i) \leq 12\abs{E}.
  \end{equation}
\end{corollary}

\begin{proof}
  Let $\cD$ denote a domain and $\cB$ a branch of the domain tree.
  We will (ab)use the notation $e\in\cD$ to denote that $e$ is adjacent to $\cD$, and $e\in\cB$ to denote that $\cB$ contains at least one domain adjacent to $e$.

  By \cref{adjacent_intersection}, if an $n$-border in $\cD$ intersects an edge of length $\leq 4^n$, then either it also intersects a defect, so is not counted in $N_E$ by definition.
  Or it also intersects an edge of length $\leq 4^n$ adjacent to $\cD$.
  Therefore
  \begin{equation}
    N_E(n,\cD^{(n)}_i,R^{(n)}_i) \leq \sum_{e\in\cD^{(n)}_i} N_E(n,\cD^{(n)}_i,R^{(n)}).
  \end{equation}
  Thus
  \begin{align}
    \sum_{n,i} N_E(n,\cD^{(n)}_i,R^{(n)}_i)
    &= \sum_{n,i}\sum_{e\in\cD^{(n)}_i} N_e(n,\cD^{(n)}_i,R^{(n)}_i)\\
    &= \sum_{e\in E}\sum_{\cD^{(n)}_i\ni e} N_e(n,\cD^{(n)}_i,R^{(n)}_i)\\
    &\leq \sum_{e\in E}\sum_{\cB\ni e}\sum_{\cD^{(n)}_i\in\cB} N_e(n,\cD^{(n)}_i,R^{(n)}_i)\\
    &\leq \sum_{e\in E}\sum_{\cB\ni e} 3 \label{eq:Ne_bound}
      \leq \sum_{e\in E} 10\cdot 3 = 30\abs{E},
  \end{align}
  where the inequalities in \cref{eq:Ne_bound} follow from \cref{Ne_bound,e_branches}, respectively.
\end{proof}

The following bound is instructive, but is not tight enough for our purposes.
\begin{proposition}[Weak border deficit bound]\label{weak_bound}
  Let $T$ be any tile configuration of a finite region $S\subset\Z_2$ of perimeter $L$.
  Let $D$ denote its defect set, and $G=(D,E)$ its defect graph.
  The total border deficit of $T$ is bounded by
  \begin{equation}
    \deficit(T) \leq 3\abs{D} + 33\abs{E} + L.
  \end{equation}
\end{proposition}

\begin{proof}
  By \cref{n-deficit} we have that
  \begin{align}
    \deficit(T)
    &= \max_R\sum_n \deficit_n(T,S,R)\\
    &\leq \max_R\sum_{n,i}\deficit_n(T,\cD^{(n)}_i,R) + \max_R\sum_{n,i}\deficit_n(T,\cU^{(n)}_i,R),\label{eq:deficit_bound}
  \end{align}
  where the inequality follows from \cref{n-domain_partition,deficit_union}.

  Applying \cref{domain_count} to the first term in \cref{eq:deficit_bound}, we obtain
  \begin{align}
    \max_R&\sum_{n,i}\deficit_n(T,\cD^{(n)}_i,R)\\
    &\leq \max_{\{R^{(n)}_i\}}\sum_{n,i}\Bigl( N_D(n,\cD^{(n)}_i,R^{(n)}_i) + N_E(n,\cD^{(n)}_i,R^{(n)}_i)
                                             + N_\partial(n,\cD^{(n)}_i,R^{(n)}_i) \Bigr)\\
    &\leq 2\abs{D} + 30\abs{E} + L, \label{eq:domain_deficit}
  \end{align}
  where the final inequality follows from \cref{ND_R_bound,NE_R_bound}, and trivially bounding the $N_\partial$ term by the total perimeter of the region being tiled.

  Applying \cref{undomain_count} to the second term in \cref{eq:deficit_bound}, we have
  \begin{align}
    \max_R&\sum_{n,i}\deficit_n(T,\cU^{(n)}_i,R)\\
    &\leq \max_R\sum_{n,i}\Bigl( N_D(n,\cU^{(n)}_i,R) + N_E(n,\cU^{(n)}_i,R) \Bigr)\\
    &\leq \max_R\Bigl( \sum_n N_D(n,\bigcup_i\cU^{(n)}_i,R) + \sum_n N_E(n,\bigcup_i\cU^{(n)}_i,R) \Bigr)
      \label{eq:U_are_disjoint}\\
    &\leq \max_R\Bigl( \sum_n N_D(n,S,R) + \sum_n N_E(n,S,R) \Bigr)
      \label{eq:expand_region}\\
    &\leq \abs{D} + 3\abs{E}. \label{eq:undomain_deficit}
  \end{align}
  In \cref{eq:U_are_disjoint}, we have used the fact from \cref{n-domain_partition} that $n$-undomains are disjoint.
  In \cref{eq:expand_region} we have used the obvious fact
  that expanding the region of consideration in $N_{D/E}$ cannot decrease these quantities.
  The final inequality follows by \cref{ND_bound,NE_bound}.

  Putting \cref{eq:domain_deficit,eq:undomain_deficit} together with \cref{eq:deficit_bound} gives the claimed bound.
\end{proof}

We would like a bound on the border deficit that scales as $O(\abs{D})$, whereas in \cref{weak_bound} $\abs{E}=\abs{D}^2$.
However, the above bound over-counts significantly, because a border that intersects some edge in $E$ will also intersect many other edges in $E$.
In the following, we tighten the bound by showing that it is sufficient to only count edges from a suitably chosen sparse subgraph of $G$.

\subsubsection{Tighter border deficit bound}

\begin{definition}[$n$-frames]\label{def:n-frame}
  For an $n$-border with corners at coordinates $(i,j)$, $(i+4^n-1,j+4^n-1)$, we define its associated \keyword{$n$-frames} to be the five $(4^n-1)\times (4^n-1)$ squares formed by the border itself, and by the four squares of the same size directly adjacent to the border,
  with corners at coordinates: $(i,j)$, $(i+4^n-1,j+4^n-1)$; $(i,j-4^n+1)$, $(i+4^n-1,j)$; $(i-4^n+1,j)$, $(i,j+4^n-1)$; $(i,j+4^n-1)$, $(i+4^n-1,j+2\cdot (4^n-1))$; $(i+4^n-1,j)$, $(i+2\cdot (4^n-1),j+4^n-1)$.

  We say that an edge \keyword{cuts} an associated $n$-frame if it has one vertex within the $n$-frame, and the other outside it.
\end{definition}

\begin{lemma}\label{intersection-implies-cutting}
  Let $G=(D,E)$ be a defect graph.
  If an edge $e\in E$ of length $\leq 4^n$ intersects an $n$-border, then it cuts at least one of its associated $n$-frames.
\end{lemma}

\begin{proof}
  An edge with both vertices in the interior of an $n$-border cannot intersect that border.
  If one of the edge's vertices is inside the $n$-border and the other outside, then it necessarily cuts the $n$-frame corresponding to the border itself.

  So consider an edge with both vertices outside the $n$-border.
  If the edge has one vertex in the interior of an $n$-frame, and the other vertex outside of that $n$-frame, then it necessarily cuts that $n$-frame.
  Thus if the edge does not cut any $n$-frame, it must either have both vertices within the interior of the same associated $n$-frame, or both vertices outside of any associated $n$-frame.
  But in both these cases, an edge of length $\leq 4^n$ cannot intersect the associated $n$-border in the first place.
\end{proof}

We denote by $N^F_E(n,S,R)$ the number of $n$-frames in a region $S$ of a Robinson tiling $R$ that are cut by an edge of length $\leq 4^n$.
The following bound follows immediately from \cref{intersection-implies-cutting}, together with the fact that by \cref{def:n-frame} an $n$-frame can be associated with at most two different $n$-borders in a Robinson tiling.

\begin{corollary}\label{NF}
  \begin{equation}
    N_E(n,S,R) \leq 2N^F_E(n,S,R).
  \end{equation}
\end{corollary}

\begin{lemma}\label{edge-cutting-bound}
  An edge of length $\leq 4^l$ can cut at most $\lfloor 4^{l-m}/2\rfloor + 6$ $n$-frames with $n\geq m$ in a Robinson tiling.
\end{lemma}

\begin{proof}
  The edge's two vertices may be located within different $n$-frames, thus the edge can cut up to 2 $n$-frames.

  Boundaries of $m$-frames with $m>n$ in a Robinson tiling can only occur in the gaps between $n$-borders.
  The $n$-borders in a Robinson tiling have sides of length $4^n-1$, and repeat periodically separated by distance $4^n$.
  Thus there are at most $2 + \lfloor 4^l/2\cdot 4^n\rfloor$ such gaps along a length $4^l$.
  Therefore, the edge can cross at most this many $m$-frame boundaries.
  For each such $m$-frame boundary, at most one of the edge's vertices can be in the interior of the $m$-frame ending at that boundary, thus this also upper-bounds the number of $m$-frames it cuts.
  The bound in the \namecref{edge-cutting-bound} follows.
\end{proof}

\begin{definition}[Delaunay triangulation]\label{def:Delaunay}
  A \keyword{Delaunay triangulation} $\Delta(P)$ of a set of points $P$ in $\R_2$ is a triangulation of $P$ such that no point in $P$ is inside the circumcircle of any triangle in $\Delta(P)$.

  A Delaunay triangulation of $P$ always exists unless $P$ are colinear.
  In the case of colinear $P$, in an abuse of notation we define $\Delta(P)$ to be the line graph connecting $P$.

  $\Delta(P)$ has the following properties~\cite{Lee_Schachter_1980}:
  \begin{enumerate}
  \item\label{def:Delaunay:circle}%
    If there exists a circle passing through $p_1,p_2\in P$ that does not contain any points from $P$ in its interior, then the edge $(p_1,p_2)$ is in $\Delta(P)$.
  \item\label{def:Delaunay:edges}%
    $\Delta(P)$ has at most $3\abs{P}-6$ edges.
  \end{enumerate}
\end{definition}
The second property follows from Euler's formula and the fact that $\Delta(P)$ is planar.

\begin{lemma}\label{Delaunay_intersection}
  Let $G=(D,E)$ be a defect graph, and $\Delta = (D,E_\Delta)$ a Delaunay triangulation of $D$.
  If an $n$-frame is cut by an edge $e\in E$ of length $\leq 4^n$, then it is also cut by an edge $e_\Delta\in E_\Delta$ of length $\leq 4^n$.

  Furthermore,
  \begin{enumerate}
  \item\label[part]{Delaunay_intersection:undomain} If $e$ is contained in an $n$-undomain $\cU$, then $e_\Delta$ is also contained in $\cU$.
  \item\label[part]{Delaunay_intersection:domain} If $e$ is adjacent to an $n$-domain $\cD$ with respect to $G$, then $e_\Delta$ is adjacent to $\cD$ with respect to $\Delta$.
  \end{enumerate}
\end{lemma}

\begin{proof}
  Let $e=(d_1,d_2)$ denote the edge in question.
  Note that, since $e$ cuts the $n$-frame, by \cref{def:n-frame} it must have one vertex $d_1$ within the $n$-frame, and one vertex $d_2$ outside it.

  We find $e_\Delta$ recursively.
  If $e$ is contained in $E_\Delta$, then set $e_\Delta=e$ and we are done.

  Otherwise, consider the circle $C$ with diameter $e$.
  Since $e\notin E_\Delta$, by \cref{def:Delaunay}\labelcref{def:Delaunay:circle} $C$ must contain another defect in its interior.
  Let $d\in\cD$ be the defect in the interior of $C$ closest to $e$.
  Note that the edges $e_1=(d,d_1)$ and $e_2=(d,d_2)$ must be strictly shorter than $e$, so have length $<4^n$.
  Therefore, by \cref{def:n-domain}, the region enclosed by the triangle $d,d_1,d_2$ -- and in particular the edges $e_1$ and $e_2$ -- are contained in the same $n$-undomain $\cU$ as $e$, fulfilling the requirements of \cref{Delaunay_intersection:undomain}.
  Moreover, since $d_1$ is within the $n$-frame and $d_2$ outside of it, one of the edges $e_1$ or $e_2$ must cut the $n$-frame.
  Denote this edge $e'$, and let $G'=(D,E\setminus e)$ be the subgraph with $e$ deleted.
  Note that $\Delta \subseteq G' \subset G$.

  If $e$ is adjacent to $\cD$ with respect to $G$, then by \cref{def:adjacent_edge} there must be a free vertical or horizontal path from $\cD$ to $e$ (i.e.\ a path that does not cross another defect, edge or $n$-domain).
  $d$ cannot be contained in the semicircle through which the free path runs, or the path would necessarily cross one of $e_1$ or $e_2$ before reaching $e$.
  Since $d$ is the defect in $C$ closest to $e$, edges $e_1$ and $e_2$ must be the next edges after $e$ that are crossed if a free path from $\cD$ is extended beyond $e$.
  Thus $e_1$ and $e_2$ are adjacent to $\cD$ with respect to $G'$.
  Therefore, as well as cutting the $n$-frame, $e'$ is also adjacent to $\cD$ with respect to $G'$ in this case, fulfilling the requirements of \cref{Delaunay_intersection:domain}.

  If $e'\in E_\Delta$, then we can set $e_\Delta=e'$ and we are done.
  Otherwise, we can repeat the preceding argument for $e'$ to obtain a new edge $e''\in G''$ with $\Delta\subseteq G''\subset G' \subset G$.
  Iterating this gives a strictly descending chain of finite subgraphs lower-bounded by $\Delta$. Since $\Delta$ is non-empty, this process must eventually terminate at a suitable $e_\Delta$.
\end{proof}

Fix a Delaunay triangulation $\Delta = (D,E_\Delta)$.
We denote by $N^F_{E_\Delta}(n,S,R)$ the number of $n$-frames in region $S$ of Robinson tiling $R$ that are cut by an edge of length $\leq 4^n$ in $E_\Delta$.
As usual, we write $N^F_e$ for $N^F_{\{e\}}$.
The following bound follows immediately from \cref{NF,Delaunay_intersection}.

\begin{corollary}\label{NF_Delta}
  \begin{equation}
    N_E(n,S,R) 
    \leq 2N^F_{E_\Delta}(n,S,R).
  \end{equation}
\end{corollary}

\begin{proposition}\label{NE_tight_bound}
  Let $T$ be a tile configuration, $G=(D,E)$ its defect graph, $S$ a region of the lattice, and $R$ a Robinson tiling of that region.
  Then
  \begin{equation}
    \sum_n N_E(n,S,R) \leq 36 \abs{D}.
  \end{equation}
\end{proposition}

\begin{proof}
  Fix a Delaunay triangulation $\Delta = (D,E_\Delta)$.
  For edges $e\in E_\Delta$, define
  \begin{equation}
    N_e(n) :=
    \begin{cases}
      \text{number of $n$-frames that are cut by $e$} & \text{ length } e \leq 4^n\\
      0 & \text{otherwise}
    \end{cases}
  \end{equation}
  so that
  \begin{equation}
    N^F_{E_\Delta}(n,S,R) \leq \sum_{e\in E_\Delta} N_e(n).
  \end{equation}
  By \cref{edge-cutting-bound}, an edge $e$ of length $\leq 4^n$ can cut at most 6 $n$-frames.
  Thus
  \begin{equation}
    \sum_n N_e(n) \leq 6.
  \end{equation}
  Putting this together with \cref{NF_Delta}, we have
  \begin{align}
    \sum_n N_E(n,S,R) &\leq 2 \sum_n N^F_{E_\Delta}(n,S,R) \leq 2 \sum_n\sum_{e\in E_\Delta} N_e(n)\\
    &\leq 2 \sum_{e\in E_\Delta}\sum_n N_e(n) \leq 2\sum_{e\in E_\Delta} 6 = 12 \abs{E_\Delta}.
  \end{align}
  Using the fact (from \cref{def:Delaunay}\labelcref{def:Delaunay:edges}) that $\abs{E_\Delta} < 3\abs{D}$, we arrive at the claimed bound.
\end{proof}

\begin{lemma}\label{NFe_bound}
  Let $T$ be a tile configuration, $\cB$ a branch of its domain tree, $\{R^{(n)}\}$ a Robinson-compatible set for the domains $\cD^{(n)}\in\cB$, and $e$ an edge in a Delaunay triangulation of the defect set of $T$.
  Then
  \begin{equation}
    \sum_n N^F_e(n,\cD^{(n)},R^{(n)}) \leq 6.
  \end{equation}
\end{lemma}

\begin{proof}
  The proof is very similar to that of \cref{Ne_bound}.

  By \cref{def:Robinson_compatible}, the $l$-borders in $R^{(n)}$ with $l\leq n$ are in the same locations as in a Robinson tiling.
  In particular, for all $n$, the $n$-borders -- and hence the $n$-frames -- in $R^{(n)}$ are located in the same place as some common Robinson tiling $R$.
  Thus $N^F_e(n,\cD^{(n)},R^{(n)}) = N^F_e(n,\cD^{(n)},R)$.

  Let $m$ be the smallest integer such that $e$ has length $\leq 4^m$.
  If $m>n$, then $N^F_e(n,\cD^{(n)},R) = 0$ by definition.
  Thus $\sum_n N^F_e(n,\cD^{(n)},R) = \sum_{n\geq m} N^F_e(n,\cD^{(n)},R)$.
  But by \cref{edge-cutting-bound}, an edge $e$ of length $\leq 4^m$ can cut at most 6 $n$-frames with $n\geq m$ in total.
  The bound follows.
\end{proof}

\begin{proposition}\label{NE_R_tight_bound}
  Let $T$ be a tile configuration, $G=(D,E)$ its defect graph, and $\{R^{(n)}_i\}$ a Robinson-compatible set for its domains $\{\cD^{(n)}_i\}$. Then
  \begin{equation}
    \sum_{n,i} N_E(n,\cD^{(n)}_i,R^{(n)}_i) \leq 360\abs{D}.
  \end{equation}
\end{proposition}

\begin{proof}
  By \cref{adjacent_intersection}, if an $n$-border in an $n$-domain $\cD$ is intersected by an edge in $E$ of length $\leq 4^n$, then either it is also intersected by a defect in $D$ hence is not counted in $N_E$ by definition.
  Or it is also intersected by an edge $e$ of length $\leq 4^n$ adjacent to $\cD$ with respect to $G$.
  \Cref{intersection-implies-cutting} in turn implies that $e$ cuts one of the $n$-frames associated with that $n$-border.
  But by \cref{Delaunay_intersection}, this implies the $n$-frame is also cut by an edge $e\in E_\Delta$ of length $\leq 4^n$ that is adjacent to $\cD$ with respect to $\Delta$.
  Thus, (ab)using the notation $e\in\cD$ to denote that $e$ is adjacent to $\cD$ with respect to $E_\Delta$,
  \begin{equation}
    N^F_{E_\Delta}(n,\cD^{(n)}_i,R^{(n)}_i) = \sum_{e\in\cD^{(n)}_i} N^F_e(n,\cD^{(n)}_i,R^{(n)}_i)
  \end{equation}
  where the sum is over edges in $E_\Delta$.

  The argument is now similar to \cref{NE_R_bound}:
  \begin{align}
    \sum_{n,i} N^F_{E_\Delta}(n,\cD^{(n)}_i,R^{(n)}_i)
    &= \sum_{n,i}\sum_{e\in \cD^{(n)}_i} N^F_e(n,\cD^{(n)}_i,R^{(n)}_i)\\
    &= \sum_{e\in E_\Delta}\sum_{\cD^{(n)}_i\ni e} N^F_e(n,\cD^{(n)}_i,R^{(n)}_i)\\
    &\leq \sum_{e\in E_\Delta}\sum_{\cB\ni e}\sum_{\cD^{(n)}_i\in\cB} N^F_e(n,\cD^{(n)}_i,R^{(n)}_i)\\
    &\leq \sum_{e\in E_\Delta}\sum_{\cB\ni e} 6 \label{eq:NFe_bound}
      \leq \sum_{e\in E_\Delta} 10\cdot 6 = 60\abs{E_\Delta},
  \end{align}
  where the inequalities in \cref{eq:NFe_bound} follow from \cref{NFe_bound} and \cref{e_branches}, respectively.

  By \cref{NF_Delta},
  \begin{equation}
    N_E(n,\cD^{(n)}_i,R^{(n)}_i) \leq 2 N^F_{E_\Delta}(n,\cD^{(n)}_i,R^{(n)}_i).
  \end{equation}
  Combining this with the above bound, and using the fact from \cref{def:Delaunay}\labelcref{def:Delaunay:edges} that $\abs{N_{E_\Delta}}\leq 3\abs{D}$, gives the desired result.
\end{proof}

\begin{theorem}\label{border_deficit}
  Let $T$ be a tile configuration of a finite subregion of $\Z_2$ with perimeter of length $L$.
  Let $D$ denote its defect set.
  The border deficit of $T$ is bounded by
  \begin{equation}
    \deficit(T) \leq 399\abs{D}+L.
  \end{equation}
\end{theorem}

\begin{proof}
  Exactly as for \cref{weak_bound}, but using the tighter bounds from \cref{NE_tight_bound,NE_R_tight_bound} in place of \cref{NE_bound,NE_R_bound}, respectively.
\end{proof}

\subsection{Square deficit bound}

\begin{definition}
  An \keyword{inner border} of an $n$-border in a Robinson tiling is an $m$-border (necessarily with $m\leq n$) located in the interior of the $n$-border, and not contained in the interior of any other border.
\end{definition}

\begin{definition}
  An \keyword{$n$-square} is a tile configuration of a $(4^n-1)\times (4^n-1)$ region of the lattice containing an $n$-border around the perimeter, inner $m$-borders in the same locations as in a Robinson tiling, and no other borders and no defects in the region between the $n$-border and the inner $m$-borders.

  We call the region of the lattice between an $n$-border and the locations where the inner $m$-borders would be in a Robinson tiling, including the $n$-border and the inner $m$-borders themselves, the \keyword{$n$-square region}. The \keyword{interior} of the $n$-square region is the region excluding the $n$-border and the inner $m$-border locations.
\end{definition}

The square deficit is defined analogously to the border deficit (\cref{def:border_deficit}).
\begin{definition}[Square deficit]\label{def:square_deficit}
  The \keyword{$n$-square deficit}, $\sdeficit_n(T,S,R)$, in a region $S$ of a tile configuration $T$ with respect to Robinson tiling $R$, is the (magnitude of the) difference between the total number of complete $n$-squares of $T$ within $S$, and the number of complete $n$-squares of $R$ within $S$.

  The \keyword{total square deficit}, $\sdeficit(T)$, of a tile configuration $T$ is the difference between the total number of complete squares in $T$ and the number of complete squares in a Robinson tiling of the same region, maximised over Robinson tilings.
\end{definition}

To extend the border deficit bound of \cref{border_deficit} to the square deficit, we need a simple lemma.
\begin{lemma}\label{intact_squares}
  Let $T$ be a tile configuration, and consider an $n$-border in $T$.
  If $T$ does \emph{not} contain an $n$-square corresponding to the $n$-border, then either one of its inner $m$-borders is missing, or there is a defect in the interior of the $n$-square.
\end{lemma}

\begin{proof}
  Assume first that there are no defects within the region corresponding to the $n$-square.
  Then the interior of the $n$-square is an $n$-domain that contains no defects, so by \cref{domain_tiling} it must contain the same pattern of $m$-borders for all $m\leq n$ as some Robinson tiling.
  The only such pattern consistent with the $n$-border itself has all outer borders in the correct locations, thus the $n$-square is intact.

  If there are no defects in the interior of the $n$-square, then by definition the $n$-square is intact unless at least one of its inner borders is missing.
\end{proof}

\begin{theorem}\label{square_deficit}
  Let $T$ be a tile configuration of a finite subregion of $\Z_2$ with perimeter of length $L$.
  Let $D$ denote its defect set.
  The square deficit of $T$ is bounded by
  \begin{equation}
    \sdeficit(T) \leq 799\abs{D}+2L.
  \end{equation}
\end{theorem}

\begin{proof}
  Every missing border in $T$, of which there are at most $\deficit(T)$, implies a missing square.
  In addition, by \cref{intact_squares}, any $n$-border in $T$ that is either missing an inner $m$-border or contains a defect in the interior of its $n$-square region, is missing its corresponding $n$-square.

  Let $B_n$ denote the number of missing inner $m$-borders (for all $m\leq n$) of $n$-borders, and $D_n$ denote the number of defects in the interior of an $n$-square region.
  Then the $n$-square deficit is bounded by
  \begin{equation}
    \sdeficit_n(T) \leq \deficit_n + B_n + D_n.
  \end{equation}

  An $m$-border can be an inner border for at most one $n$-border, thus $\sum B_n \leq \deficit(T)$.
  Since $n$-borders cannot overlap by \cref{no_border_overlap}, the interiors of $n$-square regions do not overlap either.
  Thus a defect can be contained in at most one such region, and $\sum_n D_n \leq \abs{D}$.

  Summing over $n$ and using \cref{border_deficit}, we have
  \begin{align}
    \sdeficit(T) &= \sum_n\sdeficit_n(T) \leq \sum_n\deficit_n + B_n + D_n \\
    &\leq 2\deficit(T) + \abs{D} \leq 799\abs{D} +2L,
  \end{align}
  as claimed.
\end{proof}
